\documentclass[12pt,letterpaper]{amsproc}

\usepackage[margin=1in]{geometry}
\usepackage{amsmath,amssymb,amsthm}
\usepackage{mathrsfs}
\usepackage{graphicx}
\usepackage{color}
\usepackage{slashbox} 

\newtheorem{theorem}{Theorem}
\newtheorem{corollary}{Corollary}
\newtheorem{lemma}{Lemma}

\theoremstyle{definition}
\newtheorem{assumption}{Assumption}
\newtheorem{remark}{Remark}
\newtheorem{example}{Example}

\newcommand{\Ref}[1]{(\ref{#1})}

\newcommand{\diagbox}{\backslashbox} 

\DeclareMathOperator*{\argmin}{arg\,min}
\DeclareMathOperator*{\Var}{Var}

\def\R{\mathbb{R}}
\def\1{\mathbf{1}}

\def\err{e}
\def\ve{\varepsilon}
\def\hve{\hat{\ve}}
\def\hW{\hat{W}}
\def\tW{\tilde{W}}
\def\shat{\hat{\sigma}}
\def\sj{\sum_{j=1}^n}
\def\sk{\sum_{k=1}^n}
\def\sm{\sum_{m=1}^n}
\def\avj{\frac{1}{n}\sj}
\def\avk{\frac{1}{n}\sk}
\def\avm{\frac{1}{n}\sm}

\def\opn{o_p(n^{-1/2})}
\def\Opn{O_p(n^{-1/2})}
\def\op{o_p(1)}

\def\Dset{\mathbb{D}}

\def\ahat{\hat{a}}
\def\rhat{\hat r}
\def\Fclass{\mathfrak{F}}
\def\fmem{\mathfrak{f}}

\theoremstyle{definition} 

\numberwithin{equation}{section}

\allowdisplaybreaks[1]

\begin{document} 

\title[Testing for heteroskedasticity]{Detecting heteroskedasticity in
  nonparametric regression using weighted empirical processes}

\author[J.\ Chown and U.U.\ M\"uller]{Justin Chown$^1$ and Ursula
  U. M\"uller$^2$}

\thanks{Corresponding author: Justin Chown
  (justin.chown@ruhr-uni-bochum.de) \\ 
$^1${\em Ruhr-Universit\"at Bochum, Fakult\"at f\"ur Mathematik,
  Lehrstuhl f\"ur Stochastik, 44780 Bochum, Germany.} \\
$^2${\em Department of Statistics, Texas A\&M University, College
  Station, TX 77843-3143, USA}}

\begin{abstract}
Heteroskedastic errors can lead to inaccurate statistical conclusions
if they are not properly handled. We introduce a test for
heteroskedasticity for the nonparametric regression model with
multiple covariates. It is based on a suitable residual-based
empirical distribution function. The residuals are constructed using
local polynomial smoothing. Our test statistic involves a ``detection
function" that can verify heteroskedasticity by exploiting just the
independence-dependence structure between the detection function and
model errors, i.e.\ we do not require a specific model of the variance
function. The procedure is asymptotically distribution free:
inferences made from it do not depend on unknown parameters. It is
consistent at the parametric (root-$n$) rate of convergence. Our results
are extended to the case of missing responses and illustrated with
simulations.
\end{abstract}

\maketitle

\bigskip



\section{Introduction} \label{intro}
When analysing data, it is common practice to first explore 
available options using various plotting techniques. For
regression models, a key tool is to construct a plot of the model
residuals in absolute value against fitted values. If there is only
one covariate, we can use a plot of the residuals in absolute value 
against that covariate. This technique helps determine if
theoretical requirements for certain statistical procedures are
satisfied, in particular whether or not the variation in the errors
remain constant across values of the covariate. This is an important
assumption that we want to examine more closely. We will therefore
consider the model with constant error variance $\sigma_0^2$, the 
{\em homoskedastic model}, 
which we can write in the form
\begin{equation*}
Y = r(X) + \ve, \quad \ve = \sigma_0\err.
\end{equation*}
The function $r$ is the regression function and $\sigma_0$ a positive
constant. We consider a response variable $Y$, a covariate {\em vector}
$X$ and assume that $X$ and the random variable $\err$ are independent,
where $\err$ has mean equal to zero and variance equal to one. 
 
When the variation in the data is not constant across the covariate
values, the {\em heteroskedastic model} is adequate:
\begin{equation} \label{hetero}
Y = r(X) + \ve, \quad \ve = \sigma(X)\err.
\end{equation}
Here $\sigma(\cdot)$ is a scale function with $E[\sigma^2(X)] =
\sigma_0^2$. Model (\ref{hetero}) contains the homoskedastic regression
model as a (degenerate) special case with $\sigma \equiv \sigma_0$, a
constant function. In order to discriminate between both
models we assume that $\sigma(\cdot)$ is non-constant in the
heteroskedastic case, i.e.\ it varies with the values of the
covariates $X$.

Testing for heteroskedasticity is of great importance: 
many procedures lead to inconsistent and inaccurate results if the
heteroskedastic model is appropriate but not properly handled. 
Consider model \Ref{hetero} with a parametric regression
function, e.g.\ linear regression with $r(X) = \vartheta^\top X$. The
ordinary least squares estimator $\hat \vartheta$ of the parameter
vector $\vartheta$, which is constructed for the homoskedastic model,
will still be consistent under heteroskedasticity. However it will be
less accurate than an estimator that puts more weight on observations
$(X,Y)$ with small variance $\sigma^2(X)$ (and less weight when the
variance is large). The estimated variance of $\hat \vartheta$ will 
be biased if the model is in fact heteroskedastic, so testing
hypotheses based on $\hat \vartheta$ may lead to invalid conclusions.

The relationship between the homoskedastic and heteroskedastic
models can be expressed in terms of statistical hypotheses:
%
$$
H_0:\, \exists~ \sigma_0 > 0,~ \sigma(\cdot) = \sigma_0
  \quad a.e.\ (G) \qquad \mbox{vs.} \qquad
H_a:\, \sigma(\cdot) \in \Sigma.
$$
Here $G$ is the distribution function of the covariates $X$ and
$\Sigma = \{\sigma \in L_2(G): \sigma(\cdot) > 0 \text{ and
  non-constant } a.e.(G) \}$ is a space of scale functions. The null
hypothesis corresponds to the homoskedastic model and the alternative
hypothesis to the heteroskedastic model. 

The tests introduced in this article are inspired by 
Stute, Xu and Zhu (2008), who propose tests for 
a parametric regression model with high-dimensional covariates against 
nonparametric alternatives, and by Koul, M\"uller and Schick (2012), who 
develop tests for linearity of a semiparametric
regression function for fully observed data and for a missing data
model. These approaches are in the spirit of Stute 
(1997), who introduces marked empirical processes to test parametric
models for the regression function in nonparametric regression with 
univariate covariates.

Our test statistics are
based on weighted empirical distribution functions of residuals. 
The form of these
statistics is strikingly simple and their associated limiting
behaviour is obtained by considering the related weighted empirical
process. We will show that our test
statistic converges with root-$n$ rate to a Brownian bridge. Hence
it is asymptotically distribution free and quantiles are available.
We consider detecting heteroskedasticity (represented by the
non-constant scale function $\sigma(\cdot)$) by using some
(non-constant) ``detection function'' $\omega(\cdot)$ in the space
$\Sigma$. To explain the idea, we consider the weighted error
distribution function 
\begin{equation*}
E\big[\omega(X)\1[\ve \leq t]\big] =
E\big[\omega(X)\1[\sigma(X)\err \leq t]\big],
 \qquad t \in \R.
\end{equation*}
If the null hypothesis is true $\ve = \sigma_0 e$ and we can write
\begin{equation*} 
E\big[\omega(X)\1[\ve \leq t]\big] =
E\big[\omega(X)\1[\sigma_0\err \leq t]\big]
 = E\big[
 E\big[\omega(X)\big]\1[\sigma_0\err \leq t]\big]
 = E\big[
 E\big[\omega(X)\big]\1[\ve \leq t]\big],
\end{equation*}
$t \in \R$.
Here we have used that under the null hypothesis the
covariates $X$ and the errors $\ve = \sigma_0\err$ 
are independent. This motivates a test based on the difference 
between the two quantities, i.e.\ on
\begin{equation}
\label{diffomega}
E\big[\big\{\omega(X) - E\big[\omega(X)\big]\big\}
 \1[\ve \leq t]\big],
 \qquad t \in \R,
\end{equation}
which is zero under $H_0$, but typically not under $H_a$
(see Remark \ref{remTnpow1} in Section \ref{fullmodel} for further details).
We can estimate the outer expectation by its empirical version, which
yields a test based on
\begin{equation*}
U_n(t) = n^{-1/2} \sj \Big\{\omega(X_j) - E\big[\omega(X_j)\big]\Big\}
 \1[\ve_j \leq t],
 \qquad t \in \R.
\end{equation*}
This is a process in the Skorohod space $D(-\infty,\infty)$. To move
this process to the more convenient space $D[-\infty,\infty]$, we
define the familiar limit $U_n(-\infty) = 0$ and the limit
\begin{equation*}
U_n(\infty) = n^{-1/2} \sj 
 \Big\{\omega(X_j) - E\big[\omega(X_j)\big]\Big\}.
\end{equation*}
Since the variance of $U_n(\infty)$ equals the variance of 
$\omega(X)$ it
is clear the asymptotic distribution of $\sup_{t \in \R}|U_n(t)|$ will
depend on $\Var\{\omega(X)\}$, which is not desirable for obtaining a
standard distribution useful for statistical inference. We therefore
standardise $U_n(t)$ and obtain the weighted empirical process 
\begin{equation*}
S_n(t) = n^{-1/2} \sj
 W_j \1[\ve_j  \leq t],
 \qquad t \in \R,
\end{equation*}
with weights
\begin{equation} \label{Wj}
W_j 
= \frac{\omega(X_j) - E\big[\omega(X_j)\big]}
{\Var[\omega(X_j)]^{1/2}}
= \frac{\omega(X_j) - E\big[\omega(X_j)\big]}
  {E\big[\big\{\omega(X_j) - E\big[\omega(X_j)\big]\big\}^2\big]^{1/2}
  },\qquad j=1,\ldots,n.
\end{equation}
The weights $W_j$ are centred to guarantee
that the tests are asymptotically distribution free. Related research on
(unweighted) residual-based empirical distribution functions typically 
provides uniform expansions involving a non-negligible stochastic drift 
parameter that includes the error density as a parameter of the underlying
distribution, i.e.\ the statistics are not distribution free 
(see e.g.\ Akritas and Van Keilegom, 2001; M\"uller, Schick and 
Wefelmeyer, 2007, 2009). 
This is in contrast to our case where the mean zero weights 
ensure that no drift emerges (see the discussion before 
Lemma \ref{lemTnModulus} in Section \ref{aux} for more details).

The process $S_n$ cannot be used for testing because it depends on
unknown quantities. Our final test statistic $T_n$ will therefore be 
based on an estimated version of
$S_n$ with the errors estimated by residuals $\hve_j = Y_j -
\rhat(X_j)$, $j = 1,\ldots,n$, from a sample of $n$ i.i.d.\
random variables $(X_1,\,Y_1),\ldots,(X_n,\,Y_n)$. Here $\rhat$ is a
suitable estimator of the regression function. 
In this article we assume a nonparametric regression model
and estimate the unknown smooth regression function $r$ using
a nonparametric function estimator; see Section \ref{fullmodel} 
for details.

When $\sigma(\cdot) \equiv \sigma_0$ is a constant function (the null
hypothesis is true), we expect the estimated process to behave like
$S_n(t)$ and exhibit a standard limiting behaviour. However, if
$\sigma(\cdot)$ is non-constant (the alternative hypothesis is true),
the residuals $\hve_j$ will estimate $\ve_j = \sigma(X_j)\err_j \neq
\sigma_0\err_j$ (and the weights $W_j$ and the errors $\ve_j = \sigma(X_j)
\err_j$ will not be independent). We expect the estimated process will
show a different limiting behaviour in this case. Note that our test
exploits just the independence--dependence structure between the
covariates and the errors. 
For this reason it is also clear that it will only work in our model, which
specifically assumes $\ve = \sigma(X) \err$, to test $H_0$ ``$\Var[Y|X]$
is constant'', and not in models $Y = r(X) + \ve$, where $\ve$ may depend 
on $X$ in a more general way.

The choice of the weights, i.e.\ of the detection function $\omega$,
is important to guarantee that the tests are powerful:
it is clear that $\omega$ must be non-constant to detect
heteroskedasticity. If the alternative hypothesis is true, it
will be advantageous to have weights that are highly
correlated with the scale function $\sigma$ to increase the power
of the test.
We explain this at the end of Section
\ref{fullmodel}, where we also
construct weights based on an estimate $\shat(\cdot)$ of $\sigma(\cdot)$.

Tests for heteroskedasticity are well studied for various regression
models. Glej\-ser (1969) forms a test using the absolute values of the
residuals of a linear regression fitted by ordinary
least-squares. White (1980) constructs an estimator of the covariance
matrix of the ordinary least-squares estimator in linear regression
and proposes a test based on this estimator. Cook and Weisberg (1983)
derive a score test for a parametric form of the scale function of the
errors in a linear regression. Eubank and Thomas (1993) study a test
for heteroskedasticity, which is related to the score test, for the
nonparametric regression model with normal errors. 
Although the last article studies a nonparametric regression
function, all these tests are parametric tests since the heteroskedacticity
is modelled parametrically. It is therefore possible that the tests have no 
power if those models are misspecified.

More recent papers that test for 
a parametric scale function $\sigma_\theta$ are Dette, Neumeyer and 
Van Keilegom (2007) and Dette and Hetzler (2009).
Dette et al.\  (2007) construct a test for nonparametric regression with 
univariate covariates based on the Kolomogorov--Smirnov and 
the Cramer--von--Mises statistic. They recommend a bootstrap approach to 
implement the tests. Dette and Hetzler (2009) construct a test for 
a (univariate) fixed design nonparametric regression model using an 
empirical process. The tests proposed in both papers converge at the 
root-$n$ rate.

Our approach is completely nonparametric (and thus more robust
than parametric approaches). Let us therefore look more closely
at some competing nonparametric approaches.
Dette and Munk
(1998) and Dette (2002) both consider nonparametric regression
with fixed design (i.e.\ random univariate or multivariate
covariates are not treated). Dette
and Munk create tests based on the $L_2$ distance between an
approximation of the variance function and $\sigma_0$;
in the 2002 paper Dette proposes a residual-based test using
kernel estimators. In both papers the test statistics are asymptotically 
normal with variance to be estimated. The convergence rate is
slower than the parametric root-$n$ rate (which is the rate of our test), 
and local alternatives of order $n^{-1/4}$ (Dette and Munk, 1998) or of 
order $n^{-1/2}h_n^{-1/4}$ (with bandwidth $h_n \to 0$) can be detected 
(Dette, 2002).

The approach in Dette (2002) is extended to the case of a
partially linear regression by You and Chen (2005) and Lin and Qu
(2012). 
The same idea is also used in Zheng (2009), who proposes a local smoothing
test for nonparametric regression, now with multivariate covariates, which
is also our scenario. The test statistic is again asymptotically normally 
distributed and requires a consistent estimator of the variance. 
The test is able to detect local alternatives of order $n^{-1/2} h_n^{-m/4}$, 
where $m$ is the dimension of the covariate vector, which is in agreement 
with the order in Dette (2002) for the case of a univariate $X$, $m=1$.
The test statistic is quite involved, using multivariate Nadaraya--Watson 
estimators.  
A wild bootstrap approach is used to implement the test since 
the normal approximation ``may not be accurate (...) in finite samples'' 
(Zheng, 2009, Section 5). Zheng's approach was used by Zhu et al.\ 
(2016). These authors use single-index type models (involving linear 
combinations of covariates) for the mean and the variance functions 
in order to handle high-dimensional covariates. Again, these
tests converge with a rate slower than root-$n$ and can only detect
local alternatives of order $n^{-1/2} h_n^{-q/4}$, where now $q$ is 
the number of linear combinations.

Let us finally have a closer look at two articles which, like our paper, use 
weighted empirical processes, namely Zhu et al.\ (2001) and Koul and Song 
(2010). Zhu et al.\  
use a Cramer--von--Mises type statistic of a marked empirical process
of (multivariate) {\em covariates} $X$ instead of univariate {\em residuals}
$\hve$, 
i.e.\ involving $\1(X \le x)$ instead of $\1(\hve \le t)$. 
Their test can detect local alternatives ``up to'' order
$n^{-1/2}$ and, for the case of univariate covariates, may be 
asymptotically distribution free (``under some condition''). 
The authors point directly
to the bootstrap to obtain suitable quantiles since the limiting distribution
of their statistic has a complicated covariance structure. 
However, the approach using marked 
empirical processes of {\em covariates} seems 
to be problematic when the covariate $X$ is 
multivariate because of possible dependencies between the components 
of $X$, as pointed out in Stute et al.\ (2008, p.\ 454).
Koul and Song (2010) check parametric models for the variance function and 
therefore are not directly comparable with our approach, but   
their tests have some desirable properties as well: they are 
distribution free (converging to a Brownian motion), 
and able to detect alternatives of order $n^{-1/2}$. 
Koul and Song's tests are based on a Khmaladze type transform of a marked 
empirical process of univariate {\em covariates}. Therefore a 
generalisation to the case of a multivariate $X$ does not seem to be 
advisable. 

Summing up, our approach is new in that we are the first to use
a completely nonparametric approach based on weighted ``marked'' 
empirical process of (univariate) residuals to test for heteroskedasticity. 
Our tests achieve the parametric rate root-$n$, which, so far, 
could only be achieved in Zhu et al.\ (2001), or if the tests involve 
some parametric component, e.g.\ a parametric model for the
variance function. Another advantage of our method is that the proposed 
tests are asymptotically distribution free (and quantiles are readily 
available), while most of the competing 
approaches require bootstrap to implement the test.

In this article we are also interested in the case when responses $Y$ are
missing at random (MAR), which we call the ``MAR model'', in order 
to distinguish it from the ``full model'', when all data are completely 
observed. In the MAR model the observations can be written as independent 
copies $(X_1,\delta_1 Y_1,\delta_1),\ldots,(X_n,\delta_n Y_n,\delta_n)$ of a
base observation $(X,\delta Y,\delta)$, where $\delta$ is an indicator
which equals one if $Y$ is observed and zero otherwise. 
Assuming that responses are {\em missing at random} means the
distribution of $\delta$ given the pair $(X,Y)$ depends only on the
covariates $X$ (which are always observed), i.e.\      
\begin{equation*}
P(\delta = 1|X,Y) = P(\delta = 1|X) = \pi(X).
\end{equation*}
This implies that $Y$ and $\delta$ are conditionally independent given
$X$. Assuming that responses are missing at random is often
reasonable; see Little and Rubin (2002, Chapter 1). Working with this
missing data model is advantageous because the missingness mechanism
is ignorable, i.e.\ $\pi(\cdot)$ can be estimated. It is therefore
possible to draw valid statistical conclusions without auxiliary
information, in contrast to the model with data that are ``not
missing at random'' (NMAR). Note how the MAR model covers the full
model as a special case with all indicators $\delta$ equal to $1$,
hence $\pi(\cdot) \equiv 1$.

We will show that our test statistics $T_n$, defined
in \Ref{Tn} for the full model, and $T_{n,c}$, defined in \Ref{Tc} for
the MAR model, may be used to test for the presence of
heteroskedasticity. The subscript ``c'' indicates that our test statistic
$T_{n,c}$ uses only the completely observed data; i.e.\ we use only
observations $(X,Y)$ where $\delta$ equals one, called the {\em
  complete cases}. In particular, we use only the available residuals
$\hve_{j,c} = Y_j - \rhat_c(X_j)$, where $\rhat_c$ is a suitable complete
case estimator of the regression function $r$. Demonstrating this will
require two steps. First, we study the full model and provide the
limiting distribution of the test statistic $T_n$ under the null
hypothesis in Theorem \ref{thmTn}. Then we apply the {\em transfer principle} 
for complete case statistics (Koul et al., 
2012) to adapt the results of Theorem \ref{thmTn} to the MAR model.

Since residuals can only be computed for data $(X,Y)$ that are
completely observed, the transfer principle is useful for developing 
residual-based statistical procedures in MAR regression models.
Our proposed (residual-based) tests are asymptotically
distribution free. The transfer principle guarantees, in
this case, that the test statistic and its complete case version have
the same limiting distribution (under a mild condition). This means that
one can simply omit the incomplete cases and work with the same quantiles
as in the full model, which is desirable due to its simplicity. 

This article is structured as follows. Section \ref{fullmodel}
contains the statement of the test statistic and the asymptotic
results for the full model. Section \ref{marmodel} extends the 
results of the full model to the MAR model. Simulations in 
Section \ref{simstudy} investigate the performance of these tests. 
Technical arguments supporting the results in Section
\ref{fullmodel} are given in Section \ref{aux}. 
Section \ref{conrem} concludes the article with a discussion 
of possible extensions of the proposed methodology.


\section{Completely observed data}   \label{fullmodel}
We begin with the full model and require the following standard
condition (which guarantees good performance of nonparametric
function estimators):
\begin{assumption} \label{AssumpG}
The covariate vector $X$ is quasi-uniform on the cube $[0,1]^m$; i.e.\
$X$ has a density that is bounded and bounded away from zero on
$[0,1]^m$. 
\end{assumption}

As in M\"uller, Schick and Wefelmeyer (2009), we require the
regression function to be in the H\"older space $H(d,\gamma)$, i.e.\
it has continuous partial derivatives of order $d$ (or higher) and the
partial $d$-th derivatives are H\"older with exponent $\gamma \in (0,1]$. 
We estimate the regression function $r$ by a local polynomial smoother
$\rhat$ of degree $d$. The choice of $d$ will not
only depend on the number of derivatives of $r$, but also on the
dimension $m$ of the covariate vector. (We will need more smoothness
if $m$ is large.)  We write $F$ and $f$ for the distribution function
and the density of the errors $\sigma_0 \err$ which will have to
satisfy certain smoothness and moment conditions.

 In order to describe the local polynomial smoother, let $i =
(i_1,\ldots,i_m)$ be a multi-index and $I(d)$ be the set of
multi-indices that satisfy $i_1 + \ldots + i_m \leq d$. 
Then $\rhat$ is defined as the component $\hat \beta_0$ corresponding 
to the multi-index $0=(0,\dots,0)$ of a minimiser
\begin{equation}
\label{hatbeta}
\hat \beta  = \argmin_{\beta=(\beta_i)_{i\in I(d)}} \sj \Big\{
 Y_j - \sum_{i\in I(d)}\beta_i \psi_i \Big(\frac{X_j - x}{c_n}\Big)
 \Big\}^2 w\Big(\frac{X_j - x}{c_n}\Big),
\end{equation}
where 
\begin{equation*}
\psi_i(x) = \frac{x_1^{i_1}}{i_1!}\cdots\frac{x_m^{i_m}}{i_m!},
 \qquad x=(x_1,\dots,x_m) \in [0,1]^m,
\end{equation*}
$w(x)= w_1(x_1)\cdots w_m(x_m)$ is a product of probability densities 
with compact support, and $c_n$ is a bandwidth. A typical choice 
for $w_i$ would be the Epanechnikov or the tricube kernel.
The estimator $\rhat$ was studied in M\"uller et al.\ 
(2009), who provide a uniform expansion of an empirical distribution
function based on residuals 
$$
\hve_j = Y_j - \rhat(X_j), \quad j = 1,\ldots,n. 
$$
The proof uses results from a crucial technical lemma,
Lemma 1 in that article (written here as Lemma \ref{lemrhat} in Section
\ref{aux}), which gives important asymptotic properties of $\rhat$.
We will use these properties in Section \ref{aux} to derive the limiting
distribution of our test statistic, which is based on a weighted version of 
the empirical distribution function proposed by M\"uller et al.\ (2009).

For the full model, the test statistic is given as 
\begin{equation} \label{Tn}
T_n = \sup_{t \in \R} \Big| n^{-1/2} \sj
 \hW_j\1\big[\hve_j \leq t\big] \Big|
\end{equation}
with
\begin{equation} \label{hWj}
\hW_j = \Big\{\omega(X_j) - \avk \omega(X_k)\Big\}
\Big \slash
\Big[\avm \Big\{\omega(X_m) - \avk \omega(X_k)\Big\}^2\Big]^{1/2},
\quad \omega \in \Sigma,
\end{equation}
for $j=1,\ldots,n$. The term in absolute brackets of \Ref{Tn} 
is an approximation (under $H_0$) of the process $S_n(t)$ from
the Introduction, now with the standardised weights $W_j$ from 
\Ref{Wj} replaced by empirically estimated weights $\hat W_j$.
Recall that $\omega$ must be a non-constant function, i.e.\ 
$\omega \in \Sigma$, which is crucial to guarantee that the test 
is able to detect heteroskedasticity.
 
We arrive at our main result, the limiting distribution for the test 
statistic $T_n$ in the fully observed model. The proof is given in
Section \ref{aux}.

\begin{theorem} \label{thmTn}
Let the distribution $G$ of the covariates $X$ satisfy Assumption
\ref{AssumpG}. Suppose the regression function $r$ belongs to the
H\"older space $H(d,\gamma)$ with $s = d + \gamma > 3m/2$; the
distribution $F$ of the error variable $\sigma_0\err$ has mean zero, a
finite moment of order $\zeta > 4s/(2s - m)$ and a Lebesgue density
$f$ that is both uniformly continuous and bounded; the kernel
functions $w_1,\ldots,w_m$ used in the local polynomial smoother \Ref{hatbeta}
are $(m+2)-$times continuously differentiable and
have compact support $[-1,1]$. Let $c_n \sim
\{n\log(n)\}^{-1/(2s)}$. 
Let the null hypothesis hold. Then
\begin{equation*}
T_n = \sup_{t \in \R} \Big| n^{-1/2} \sj
 \hW_j\1\big[\hve_j \leq t\big] \Big|
\end{equation*}
with $\hat W_j$ specified in \Ref{hWj} above,
converges in distribution to $\sup_{ t \in [0,1]} |B_0(t)|$, where
$B_0$ denotes the standard Brownian bridge.
\end{theorem}
The distribution of $\sup_{t \in [0,1]} |B_0(t)|$ is a standard
distribution, whose upper $\alpha$-quantiles $b_\alpha$ can be
approximately calculated using formula (12) on page 34 of Shorack
and Wellner (1986), i.e.\
\begin{equation*}
P\Big( \sup_{t \in [0,1]} \big|B_0(t)\big| \leq b \Big) 
 = \sqrt{ \frac{2\pi}{b} } \sum_{k = 1}^{\infty}
 \exp\Big( -\frac{ ( 2k - 1 ) ^ 2 \pi^2 }{ 8 b^2 } \Big),
 \qquad b > 0.
\end{equation*}
We calculate that $b_{0.05} = 1.1779$ is an appropriate quantile for
a $5\%$ level test.

\begin{remark}[{\sc power under fixed alternatives}] \label{remTnpow1}
It is possible that the test has no power if the detection function
$\omega$ is not properly chosen. To see this consider the difference
$$
E\big[\big\{\omega(X) - E\big[\omega(X)\big]\big\} \1[\ve \leq t]\big]
$$
from equation \Ref{diffomega} in the introduction, which is zero under
$H_0$. The test has no power if the difference is also zero under $H_a$,
i.e.\ if 
$$
E\big[ \omega(X) \1[ \ve \leq t] \big]
=
E\Big[ \omega(X) F\Big(\frac t {\sigma(X)} \Big)\Big]
=
E\{\omega(X)\} E \Big[ F\Big(\frac t {\sigma(X)} \Big)\Big]
=
E\{\omega(X)\} E \big[ \1[ \ve \leq t]   \big],
$$
which means that $\omega(X)$ and $F(t/\sigma(X))$ are uncorrelated.
This happens, for example, if $X$ and $\ve$ are both uniformly distributed
on $[0,1]$ and if $\sigma(X) = [1 + \sin(2 \pi X)]^{-1}$
and $\omega(X) = 1 + \cos(2 \pi X)$.
Then $E[ \omega(X) F(t/\sigma(X))] = t E[ \omega(X)/\sigma(X)]$. It is
easy to check that this indeed equals $t E[ \omega(X)] E[1/\sigma(X)]$ 
since  $E[ \omega(X)/\sigma(X)] = E[ \omega(X)] = E[1/\sigma(X)] =1$. 
\end{remark}

\begin{remark}[{\sc power under local alternatives}] \label{remTnpower} 
To derive the power of the test under local alternatives of the form
$\sigma = \sigma_{n\Delta} = \sigma_0 + n^{-1/2}\Delta$
with $\sigma_{n\Delta} \in \Sigma$,
we use Le Cam's third lemma. This result states that a local shift
$\Delta$ away from the null hypothesis results in an additive shift of
the asymptotic distribution of $T_n$; see e.g.\ page 90 of van der
Vaart (1998). The shift is calculated as the covariance between $T_n$
and $\log(dF_{n\Delta}/dF)$ under $H_0$ where $\ve = \sigma_0\err$. 
Here $F_{n\Delta}(t) = P(\{\sigma_0 + n^{-1/2}\Delta(X)\}\err \leq t\,|\,X)$.
A brief sketch shows $E[T_n \log(dF_{n\Delta}/dF) ]$ is equal to
\begin{align*}
&E\Big[ \Big\{ n^{-1/2}\sj
 W_j\1\big[\sigma_0\err_j \leq t\big] \Big\}
 \Big\{n^{-1/2}\sj \Delta(X_j) + n^{-1/2} \sj \Delta(X_j)
 \sigma_0 \err_j \frac{f'(\sigma_0\err_j)}{f(\sigma_0\err_j)}
 \Big\} \Big] + o(1) \\
&= E[W\Delta]F(t) +
 E[W\Delta]\int_{-\infty}^t s \, \frac{f'(s)}{f(s)}\,F(ds) + o(1) \\
&= t \, f(t)E[W\Delta] + o(1).
\end{align*}
Here we have, for simplicity, assumed that
$F$ is differentiable with finite Fisher information for location and scale.
Hence, under a contiguous alternative $H_a$, the distribution of the
test statistic $T_n$ limits to 
$\sup_{t \in [0,1]}|B_0(t) + F^{-1}(t) \{f \circ F^{-1}(t)\}E[W\Delta]|$, 
writing $F^{-1}$ for the quantile function of $F$.
\end{remark}

Since the weights in our test statistic are standardised, only the
shape of $\omega$ may have an effect on the behaviour of the statistic
-- location and scale have no influence. From Remark \ref{remTnpower}
we know that the test has no power under local alternatives if $E[W
\Delta] = 0$, i.e.\ by definition of $W$, if the detection function
$\omega(X)$ and the scale function $\sigma(X) = \sigma_0 +
n^{-1/2}\Delta(X)$ are uncorrelated. This happens, for example, if $X$
has a standard uniform distribution, if $\Delta(X) = 1 + \sin(2 \pi
X)$ and if we choose $\omega(X) = 1 + \cos(2 \pi X)$; cf.\ Remark
\ref{remTnpow1}.

From Remark \ref{remTnpower} it is also clear that
the power of the test increases with $E(W \Delta)$, i.e.\ if
$\omega(X)$ and $\Delta(X)$ are highly correlated.
So it can be expected that the test will perform best when $\omega$ is a linear
transformation of the scale function $\sigma$. This suggests choosing
$\omega$ similar in shape to $\sigma$, in order to obtain a powerful test. 
We propose using $\omega = \hat \sigma$, where $\hat \sigma$ is
a consistent estimator of $\sigma$. 
Assume for simplicity that the regression function $r$ and the second
conditional moment $r_2(\cdot) = E[Y^2\,|\,X = (\cdot)]$ of $Y$ given
$X$ belong to the same H\"older class $H(d,\gamma)$ with $s = d +
\gamma$. Then $r_2$ can be estimated by a local polynomial smoother
$\rhat_2$, which is defined as $\rhat$ in \eqref{hatbeta} but now
with $Y_j^2$ in place of $Y_j$, $j = 1,\ldots,n$. This leads to an
estimator of $\sigma$ using $\shat = \{\rhat_2 -
{\rhat}^2\}^{1/2}$. The estimated weights are then given by
\begin{equation} \label{tWj}
\tilde W_j = \Big\{\hat \sigma(X_j) - \avk\hat \sigma(X_k)\Big\}
 \Big\slash\Big[
 \avm\Big\{\hat \sigma(X_m) - \avk\hat \sigma(X_k)\Big\}^2
 \Big]^{1/2}
\end{equation}
for $j=1,\ldots,n$.

The formal result for this choice of weights is given in Theorem 
\ref{thmTnWhat} below.

Note that the weights in \Ref{tWj} are non-degenerate under the
null hypothesis because the terms in the numerator and in the denominator 
have the same asymptotic order. Since the statistic converges weakly to 
a Gaussian process, which is determined by its mean and covariance functions, 
and since our weights are scaled and centred, they do not affect the
asymptotic distribution.

\begin{theorem} \label{thmTnWhat}
Suppose the assumptions of Theorem \ref{thmTn} are satisfied with the 
error variable $\sigma_0e$ having a finite moment of order larger than
$8$. Assume that $r$ and $r_2$ belong to the H\"older space
$H(d,\gamma)$ with $s = d + \gamma > 3m/2$. Then under the null
hypothesis
\begin{equation*}
\tilde T_n = \sup_{t \in \R}\Big|
 n^{-1/2} \sj \tilde W_j\1\big[\hve_j \leq t\big] \Big|
\end{equation*}
with $\tilde W_j$ specified in \Ref{tWj} above,
converges in distribution to $\sup_{t \in \R}|B_0(t)|$, where $B_0$
denotes the standard Brownian bridge. In addition, $\tilde T_n$
consistently detects alternative hypotheses $\sigma \in H(d,\gamma)$,
and $\tilde T_n$ is asymptotically most powerful for detecting local
alternative hypotheses of the form $\sigma_0 + n^{-1/2}\Delta(\cdot)$,
where $\Delta \in H(d,\gamma)$.
\end{theorem}
The first part of Theorem \ref{thmTnWhat} is verified in the
supplementary online materials. The last statement concerning the
power of the test follows from Remark \ref{remTnpower} and the 
discussion afterwards, combined with
the fact that we can consistently estimate the scale provided $\Delta
\in H(d,\gamma)$. 


\section{Responses missing at random} \label{marmodel}
We now consider the MAR model. The complete case test statistic is
given by
\begin{align} \label{Tc}
T_{n,c} 
= \sup_{t \in \R} \Big| N^{-1/2} \sj \delta_j \hW_{j,c}
 \1\big[\hve_{j,c} \leq t\big] \Big|, 
\quad \mbox{ with } \, \hve_{j,c} = Y_j - \rhat_c(X_j).
\end{align}
Here $N = \sj \delta_j$ is the number of complete cases and
$\hW_{j,c}$ denotes the weights from equation \Ref{hWj} in the previous 
section, which are now constructed using only the complete cases. The estimator 
$\rhat_c$ is the complete case version of $\rhat$; i.e.\ the component 
$\hat \beta_{c,0}$ corresponding to the multi-index $0=(0,\dots,0)$ of 
a minimiser
\begin{equation*}
\hat \beta_c = \argmin_{\beta=(\beta_i)_{i\in I(d)}} \sj
 \delta_j\Big\{
 Y_j- \sum_{i\in I(d)} \beta_i \psi_i \Big(\frac{X_j - x}{c_n}\Big)
 \Big\}^2 w\Big(\frac{X_j - x}{c_n}\Big),
\end{equation*} 
which is defined as in the previous section, but now also involves 
the indicator $\delta_j$. 

The transfer principle for complete case statistics (Koul et al.,
2012) states that if the limiting distribution of a statistic in the 
full model is $\mathcal L(Q)$, with $Q$ the joint distribution of
$(X,Y)$, then the distribution of its complete case version in the MAR
model will be $\mathcal L(Q_1)$, where $Q_1$ is the conditional
distribution of $(X,Y)$ given $\delta = 1$. The implication holds
provided $Q_1$ satisfies the same model assumptions as $Q$. For our
problem this means that $Q_1$ must meet the assumptions imposed on
$Q$ by Theorem \ref{thmTn}. It is easy to see how this affects only
the covariate distribution $G$. Due to the independence of $X$ and
$\err$, the distribution $Q$ of $(X,Y)$ factors into the marginal 
distribution $G$ of $X$ and the conditional distribution of $Y$ 
given $X$, i.e.\ the distribution $F$ of the errors $\sigma_0\err$. 
This means we can write $Q = G \otimes F$. 
The MAR assumption implies that $\err$ and $\delta$ are independent.  
Hence the distribution $F$ of the errors remains unaffected when we
move from $Q$ to the conditional distribution $Q_1$ given $\delta = 1$,
and we have $Q_1 = G_1 \otimes F$, where $G_1$ is the distribution
of $X$ given $\delta=1$.
Thus, Assumption \ref{AssumpG} about $G$ must be
restated; we also have to assume the detection function $\omega$ is
square-integrable with respect to $G_1$.
\begin{assumption} \label{AssumpG1}
The conditional distribution $G_1$ of the covariate vector $X$ given
$\delta = 1$ is quasi-uniform on the cube $[0,1]^m$; i.e.\ it has a
density that is bounded and bounded away from zero on $[0,1]^m$.
\end{assumption}

The limiting distribution $\mathcal L(Q)$ of the test statistic in the full 
model in Theorem \ref{thmTn} is given by $\sup_{t \in [0,1]}
|B_0(t)|$. Hence it does {\em not} depend on the joint distribution $Q$
of $(X,Y)$ (or on unknown parameters). This makes
the test particularly interesting for the MAR model, since the limiting 
distribution of the complete case statistic $\mathcal L(Q_1)$ is the same 
as the distribution of the full model statistic,  
$\mathcal L(Q_1) = \mathcal L(Q)$,
i.e.\ it is also given by $\sup_{t \in [0,1]}|B_0(t)|$.
Combining these arguments already provides proof for the
main result for the MAR model.

\begin{theorem} \label{thmTc}
Let the null hypothesis hold. Suppose the assumptions of Theorem
\ref{thmTn} are satisfied, with Assumption \ref{AssumpG1} in place of
Assumption \ref{AssumpG}, and let $\omega \in L_2(G_1)$ be positive
and non-constant. Write
\begin{equation*}
\hW_{j,c} = \Big\{ \delta_j\omega(X_j)
 - \frac{1}{N}\sum_{k=1}^n \delta_k\omega(X_k)\Big\}
 \Big\slash \Big[
 \frac{1}{N}\sm\Big\{\delta_m\omega(X_m)
 - \frac{1}{N}\sum_{k=1}^n \delta_k\omega(X_k)\Big\}^2
 \Big]^{1/2}
\end{equation*}
and $\hve_{j,c} = Y_j - \rhat_c(X_j)$. Then
\begin{equation*}
T_{n,c} = \sup_{t \in \R} \Big| N^{-1/2}\sj 
 \delta_j\hW_{j,c}\1\big[\hve_{j,c} \leq t\big] \Big| 
\end{equation*}
converges in distribution to $\sup_{t \in [0,1]} |B_0(t)|$, where
$B_0$ denotes the standard Brownian bridge.
\end{theorem}

This result is very useful: if the assumptions of the MAR model are
satisfied it allows us to simply delete the incomplete cases and
implement the test for the full model; i.e.\ we may use the same
quantiles.

\begin{remark} \label{remTcWhat}
Following the discussions above and those preceding Theorem 
\ref{thmTnWhat} in the previous section we can construct estimated 
weights based on complete cases as follows. The first and second 
conditional moments of $Y$ given $X$ can be estimated by complete case 
versions $\hat r_{c}$ and $\hat r_{2,c}$ of the local polynomial smoothers
$\hat r$ and $\hat r_{2}$. 
Hence, $\hat \sigma_c(\cdot) =
\{\hat r_{2,c}(\cdot) - \hat r_{c}^2(\cdot)\}^{1/2}$ 
is a consistent
complete case estimator of $\omega(\cdot) = \sigma(\cdot)$ (which 
optimises the power of the test). The complete case version of the 
test statistic 
$\tilde T_n$ is
\begin{equation*}
\tilde T_{n,c} = \sup_{t \in \R}\Big|
 N^{-1/2} \sj \delta_j\tilde W_{j,c}\1\big[\hve_{j,c} \leq t\big] \Big|,
\end{equation*}
where the weights $\tilde W_{j,c}$ are complete case versions of $\tilde W_j$;
see \Ref{tWj}.
The transfer principle then implies that the results of Theorem
\ref{thmTnWhat} continue to hold for $\tilde T_{n,c}$, i.e.\
$\tilde T_{n,c}$ tends under the null hypothesis in distribution to 
$\sup_{t \in [0,1]} |B_0(t)|$ and is asymptotically most powerful
for detecting smooth local alternatives.
\end{remark}


\section{Simulation results} \label{simstudy}
A brief simulation study demonstrates the effectiveness of a
hypothesis test using the test statistics given above for the full
model and the MAR model.

The test statistics for the full model and the MAR model are based on
the nonparametric estimator $\rhat$ (see \Ref{hatbeta}), which
involves a bandwidth $c_n = c \, \{n \log(n)\}^{-1/(2s)}$, with
proportionality constant $c$ that has to be suitably chosen. We
recommend selecting $c_n$ (and thus $c$) by cross-validation, i.e.\
$c_n$ is the bandwidth that minimises the leave-one-out
cross-validated estimate of the mean squared prediction error (see
e.g.\ H\"ardle and Marron, 1985). This procedure is easy to implement,
is fully data-driven and performed well in our study.
It can also be used in the scale
function estimator $\hat \sigma$ for a test based on $\tilde T_n$,
which is what we did in the examples below.

The test level is $\alpha = 5\%$ in the following and the asymptotic
quantile (introduced after Theorem \ref{thmTn}) is therefore $b_{0.05}
\approx 1.1779$. For small to moderate sample sizes we recommend the
smooth residual bootstrap approach by Neumeyer (2009) for estimating
the quantiles, which worked well in our simulation study. In
particular, when the sample sizes were small ($50$ or less) the
results using the asymptotic quantile $b_{0.05}$ were not satisfactory
in general and the results using the bootstrap quantile were more
plausible. At moderate and larger sample sizes the bootstrap
quantiles and the asymptotic quantile $b_{0.05}$ produced similar
results.

The bootstrap method is suitable for our setting as it makes it possible to
produce a smooth bootstrap distribution that satisfies our model
assumptions. Note that the smooth bootstrap approach is based on
residuals $\hve = Y - \rhat(X)$ and the estimated weights $\tW_j$ from
\eqref{tWj}, i.e.\ it also involves the bandwidths selected by
cross-validation that were introduced above.

\medskip\noindent
{\bf Example 1: testing for heteroskedasticity with one covariate.}
For the simulations we chose the regression function as
\begin{equation*}
r(x) = 2x + 3\cos(\pi x)
\end{equation*}
to preserve the nonparametric nature of the model. The covariates were
generated from a uniform distribution and errors from a standard
normal distribution: $X_j \sim U(-1,1)$ and $\err_j \sim N(0,1)$ for
$j = 1,\ldots,n$. Finally, the indicators $\delta_j$ have a
Bernoulli$(\pi(x))$ distribution, with $\pi(x) = P(\delta = 1| X =
x)$. In this study we use a logistic distribution function for
$\pi(x)$ with a mean of 0 and a scale parameter of 1. As a
consequence the average amount of missing data is around 50\%,
ranging between 27\% and 73\%. We work with $d=1$, the locally linear
smoother and sample sizes $50$, $100$, $200$ and $300$.


\begin{figure}
\centering
\begin{minipage}[c]{0.95\textwidth}
\centering
\includegraphics[width=0.475\textwidth]{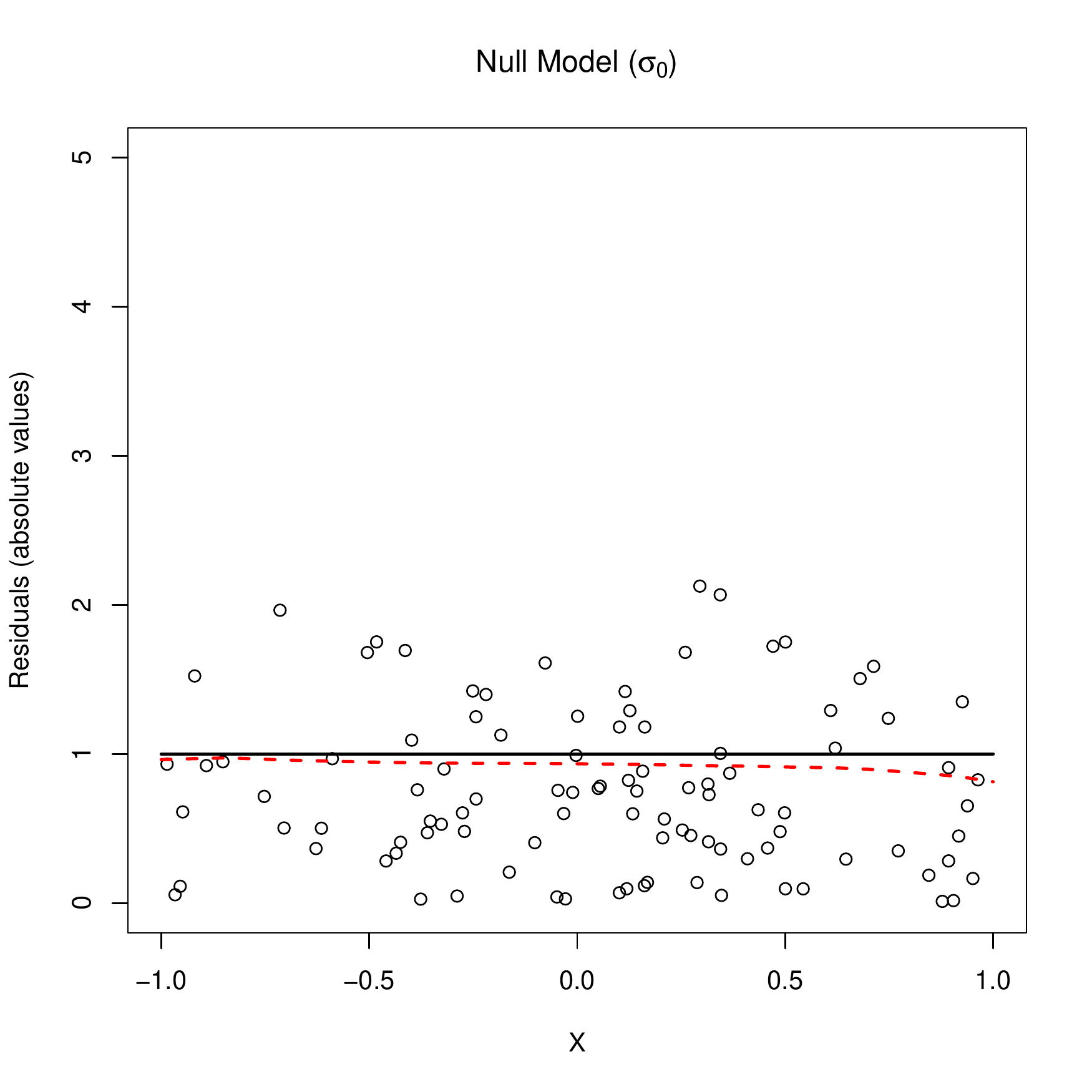}
\includegraphics[width=0.475\textwidth]{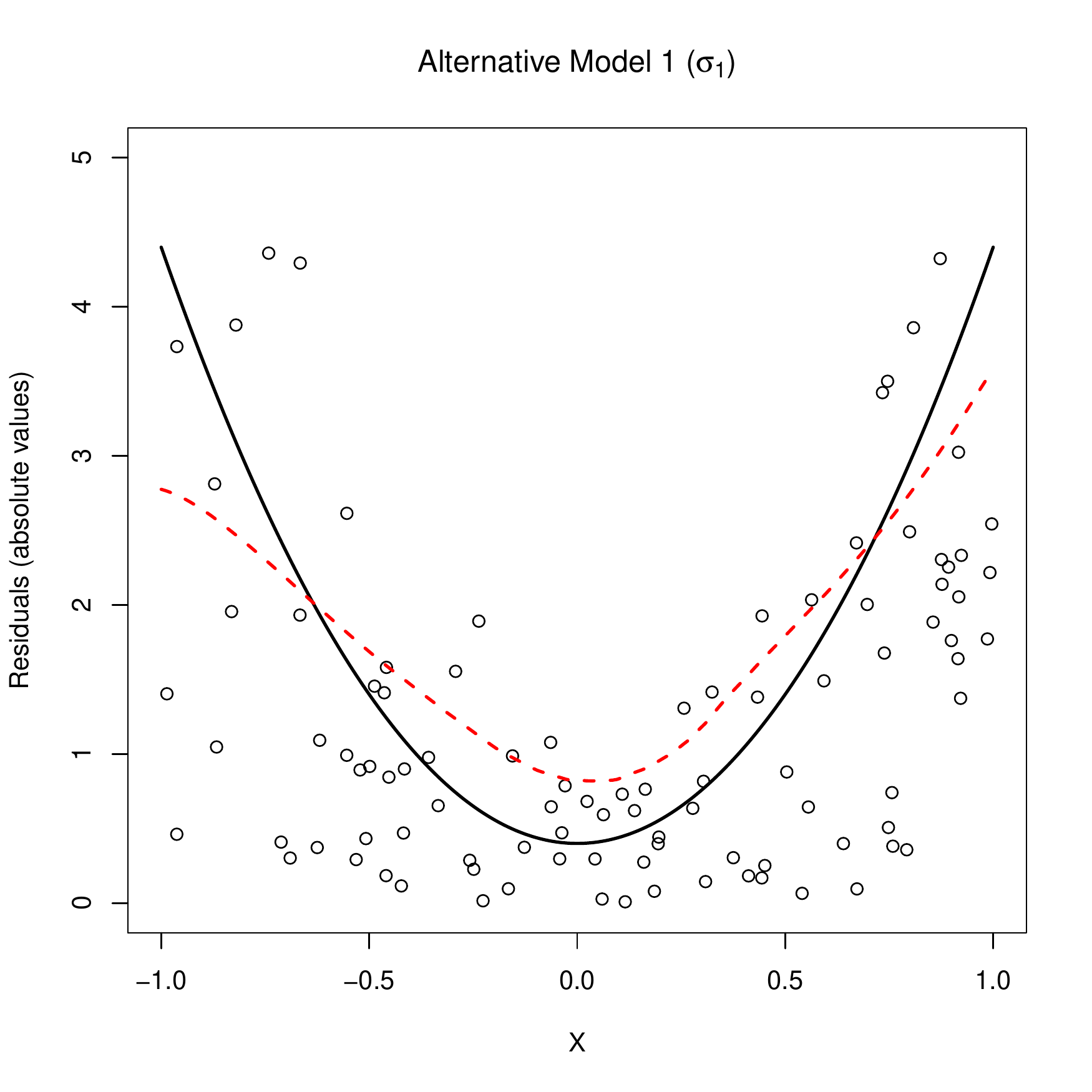}
\end{minipage}
\begin{minipage}[c]{0.95\textwidth}
\centering
\includegraphics[width=0.475\textwidth]{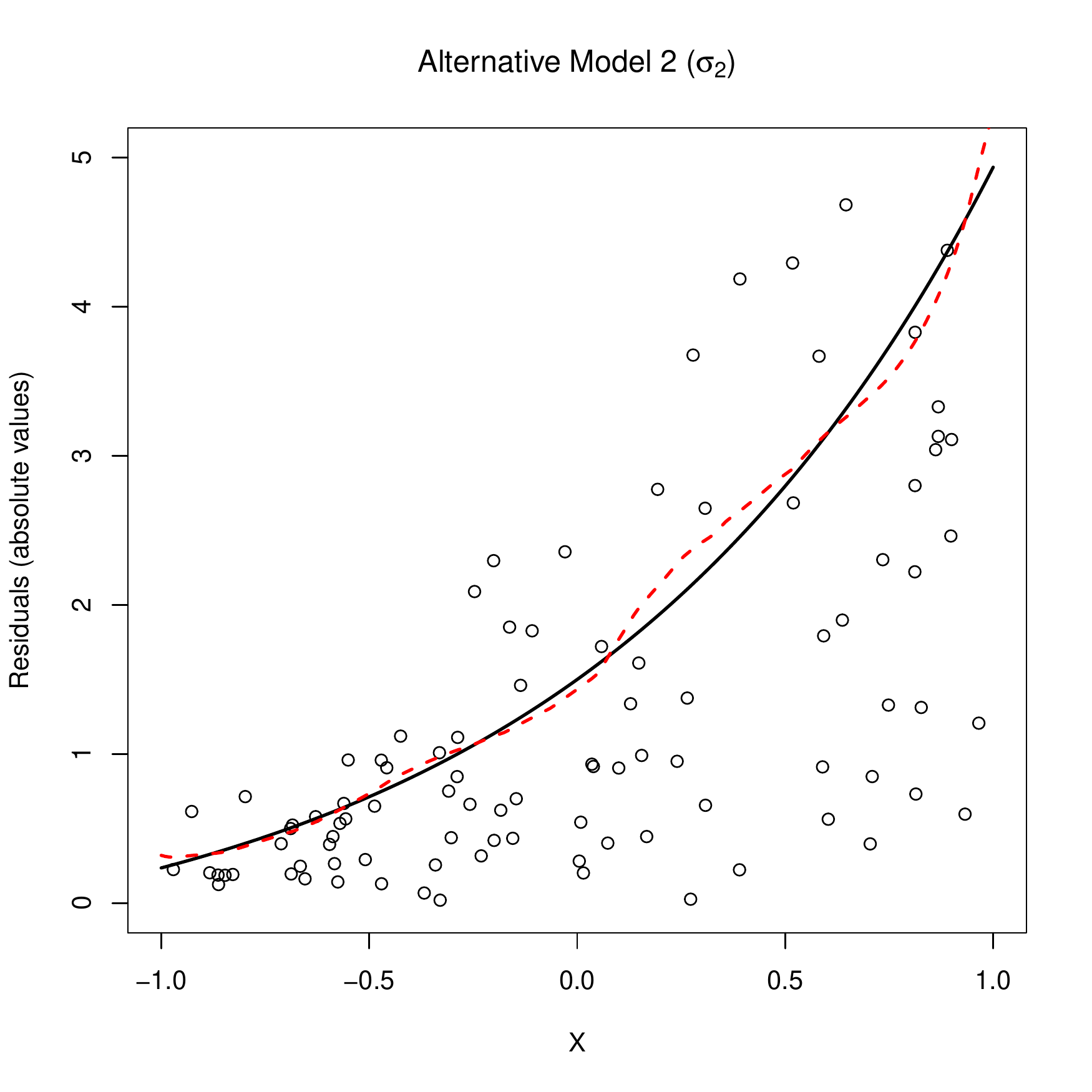}
\includegraphics[width=0.475\textwidth]{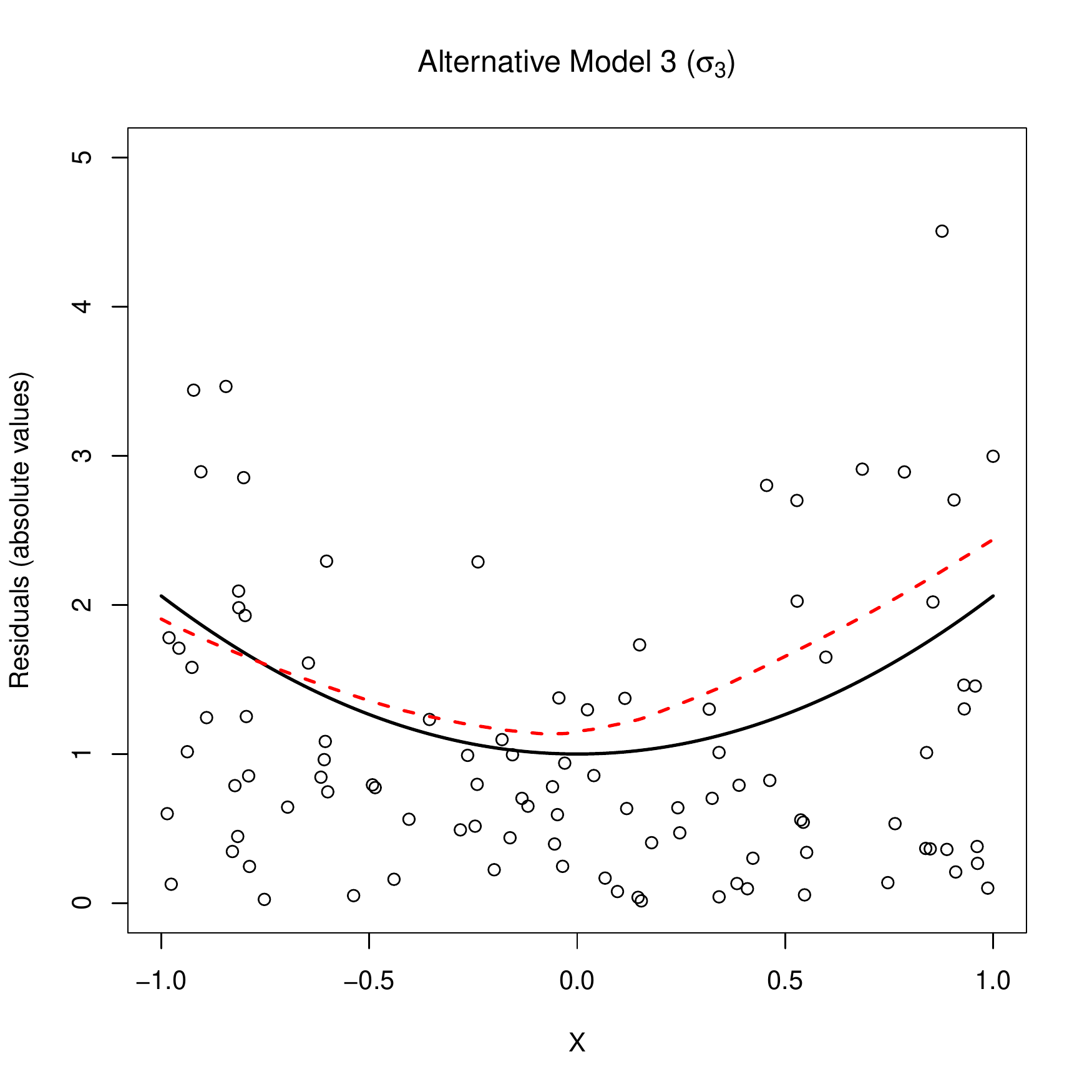}
\end{minipage}
\caption{Scatter plots of absolute valued residuals. Each plot
also shows the underlying scale function in black (solid line) and a kernel
smoothed estimate of the scale function in red (dashed line).}
\label{figure1}
\end{figure}


In order to investigate the level and power of the test in the full
model and in the MAR model we consider the following scale functions:
\begin{gather*}
\sigma_0(x) = 1,  \qquad
\sigma_1(x) = 0.4 + 4x^2, \\
\sigma_2(x) = 2e^{x} - 0.5, \qquad
\sigma_3(x) = 1 + 15n^{-1/2}x^2.
\end{gather*}
The constant scale function $\sigma_0$ allows for the (5\%) level of
the test to be checked.
As an illustration, we generated a random dataset of size 100 for each
scenario. A scatter plot of the residuals (in absolute value) from the
nonparametric regression is given for each dataset (Figure
\ref{figure1}).

The simulations based on (non-constant) scale
functions $\sigma_1$, $\sigma_2$ and $\sigma_3$ give an indication of
the power of the test in different scenarios. In particular, we
consider the power of the test against the local alternative
$\sigma_3$.
The power is maximised if $\omega$ equals the scale
function $\sigma$ (or is a linear transformation of $\sigma$); see
Remark \ref{remTnpower} in Section \ref{fullmodel} and the discussion
following it.
The formula for the weights based on an estimator
$\hat \sigma(x) = \{\hat r_2(x) - \hat r^2(x)\}^{1/2}$ of $\sigma$ is given
in \Ref{tWj}.

To check the performance of our test we conducted simulations
of 1000 runs. Table \ref{tableTn} shows the test results for fully
observed data ($\tilde T_n$). Similar results are given for missing data
($\tilde T_{n,c}$) in Table \ref{tableTnc}.

The figures corresponding to the null hypothesis ($\sigma_0$) show that
test levels for fully observed data ($\tilde T_n$) are near the
desired $5\%$ at larger sample sizes ($3.9\%$ at sample size $300$)
but more conservative at smaller sample sizes ($1.6\%$ at sample size
$50$).
The results for missing data ($\tilde T_{n,c}$) are further away from
$5\%$ when the sample size is small and the asymptotic quantile
$b_{0.05}$ is used. Both testing procedures appear to have levels near
the desired $5\%$ when the smooth bootstrap quantiles are used, which
is expected.


\begin{table}
\centering
\small
\begin{tabular}{|r|c c c c|}
\hline
\diagbox{$\sigma$}{$n$} & $50$ & $100$ & $200$ & $300$ \\
\hline
$\sigma_0$ & $0.016$ ($0.058$) & $0.019$ ($0.056$) &
 $0.033$ ($0.039$) & $0.039$ \\
$\sigma_1$ & $0.426$ ($0.477$) & $0.939$ ($0.945$) &
 $1.000$ ($1.000$) & $1.000$ \\
$\sigma_2$ & $0.487$ ($0.631$) & $0.971$ ($0.957$) &
 $1.000$ ($0.996$) & $1.000$ \\
$\sigma_3$ & $0.127$ ($0.176$) & $0.299$ ($0.387$) &
 $0.500$ ($0.576$) & $0.668$ \\
\hline
\end{tabular}
\vspace{1ex}
\caption{Example 1: Simulated level ($\sigma_0$ figures) and power for
  fully observed data ($\tilde T_n$). Figures in parentheses are
  obtained from bootstrap quantiles.}
\label{tableTn}
\end{table}



\begin{table}
\centering
\small
\begin{tabular}{|r|c c c c|}
\hline
\diagbox{$\sigma$}{$n$} & $50$ & $100$ & $200$ & $300$ \\
\hline
$\sigma_0$ & $0.009$ ($0.054$) & $0.015$ ($0.051$) &
 $0.029$ ($0.048$) & $0.037$ \\
$\sigma_1$ & $0.097$ ($0.173$) & $0.482$ ($0.573$) &
 $0.957$ ($0.953$) & $0.998$ \\
$\sigma_2$ & $0.112$ ($0.223$) & $0.443$ ($0.550$) &
 $0.945$ ($0.913$) & $1.000$ \\
$\sigma_3$ & $0.032$ ($0.080$) & $0.097$ ($0.148$) &
 $0.197$ ($0.283$) & $0.304$ \\
\hline
\end{tabular}
\vspace{1ex}
\caption{Example 1: Simulated level ($\sigma_0$ figures) and power for
  missing data ($\tilde T_{n,c}$). Figures in parentheses are
  obtained from bootstrap quantiles.}
\label{tableTnc}
\end{table}


We now consider the power of each test
beginning with $\sigma_1$. The procedure for fully observed data
($\tilde T_n$) performs very well at moderate and larger sample
sizes. For example, we rejected the null hypothesis $93.9\%$ of the
time at the moderate sample size of $100$ using the asymptotic
quantile. In this case, results using the smooth bootstrap quantile
were almost identical (rejecting $94.5\%$ of the time).
Similar results were obtained for missing data ($\tilde T_{n,c}$), but
they are (as expected) less impressive.
Note that the smooth bootstrap quantiles do not (in general) perform
dramatically better than the asymptotic quantile $b_{0.05}$.

The figures corresponding to $\sigma_2$ and $\sigma_3$ show that both
tests have difficulty rejecting when samples are small. The
procedure for fully observed data ($\tilde T_n$) only rejected the
null hypothesis $48.7\%$ ($\sigma_2$) and $12.7\%$ ($\sigma_3$) of the
time for samples of size $50$ and less often when data were
missing. Here the smooth bootstrap quantiles show improved
performance over the asymptotic quantile $b_{0.05}$ and reject the
null hypothesis $63.1\%$ ($\sigma_2$) and $17.6\%$ ($\sigma_3$) of the
time. The results are similar for missing data but with reduced
performance (as expected).

In conclusion, each test performs well and the procedures $\tilde T_n$
and $\tilde T_{n,c}$ proposed in this article appear promising for
detecting heteroskedasticity. When sample sizes are small the smooth
bootstrap quantiles appear to be helpful.

\medskip\noindent
{\bf Example 2: testing for heteroskedasticity with two covariates.}
Here we work with the regression function
\begin{equation*}
r(x_1,x_2) = 2x_1 - x_2 + 3e^{x_1 + x_2}.
\end{equation*}
The covariates $X_1$ and $X_2$ are generated from a
joint normal distribution, with component variances $1$ and correlation
coefficient $1/2$, restricted to the interval $[-1, 1]^2$ by
rejection sampling. As above we generate the model errors from a
standard normal distribution. In this example we do not consider
missing data because we expect the conclusions to mirror those of the
first simulation study. We are interested in the performance of our
testing procedure when we select different weights. We work with
$d=3$, the locally cubic smoother, and sample sizes $50$, $100$, $200$
and $300$. The level of the test is $5\%$, as in Example 1.
In Example 1 the bootstrap and the test based on the asymptotic quantile
$b_{0.05}$ produced similar results for sample sizes 100 and larger.
We therefore only consider the latter method in this example.


\begin{table}
\centering
\small
\begin{tabular}{|r|c c c c|}
\hline
\diagbox{$\sigma$}{$n$} & $50$ & $100$ & $200$ & $300$ \\
\hline
$\sigma_0$ & $0.017$ & $0.015$ & $0.016$ & $0.019$ \\
$\sigma_1$ & $0.153$ & $0.749$ & $0.996$ & $1.000$ \\
$\sigma_2$ & $0.027$ & $0.023$ & $0.025$ & $0.026$ \\
\hline
\end{tabular}
\vspace{1ex}
\caption{Example 2: Simulated level ($\sigma_0$ figures) and power for
  $\tilde T_n$ using detection function $\omega_1 = \sigma_1$.}
\label{table2w1}
\end{table}
\begin{table}
\centering
\small
\begin{tabular}{|r|c c c c|}
\hline
\diagbox{$\sigma$}{$n$} & $50$ & $100$ & $200$ & $300$ \\
\hline
$\sigma_0$ & $0.003$ & $0.010$ & $0.018$ & $0.019$ \\
$\sigma_1$ & $0.085$ & $0.415$ & $0.904$ & $0.989$ \\
$\sigma_2$ & $0.015$ & $0.037$ & $0.118$ & $0.208$ \\
\hline
\end{tabular}
\vspace{1ex}
\caption{Example 2: Simulated level ($\sigma_0$ figures) and power for
  $\tilde T_n$ using detection function $\omega_2 = 1 +
  \cos((\pi/2)(x_1 + x_2))$. }
\label{table2w2}
\end{table}
\begin{table}
\centering
\small
\begin{tabular}{|r|c c c c|}
\hline
\diagbox{$\sigma$}{$n$} & $50$ & $100$ & $200$ & $300$ \\
\hline
$\sigma_0$ & $0.063$ & $0.049$ & $0.060$ & $0.069$ \\
$\sigma_1$ & $0.186$ & $0.617$ & $0.971$ & $1.000$ \\
$\sigma_2$ & $0.477$ & $0.934$ & $1.000$ & $1.000$ \\
\hline
\end{tabular}
\vspace{1ex}
\caption{Example 2: Simulated level ($\sigma_0$ figures) and power for
  $\tilde T_n$ using detection function $\omega_3$ as an estimated
  scale function.}
\label{table2w3}
\end{table}

For the simulations we use three scale functions:
\begin{gather*}
\sigma_0 \equiv 1,
\quad
\sigma_1(x_1,x_2) = 0.5 + 5x_1^2 + 5x_2^2,
\quad
\sigma_2(x_1,x_2) = 4 + 3.5 \sin\big((\pi/2)(x_1 + x_2)\big)
\end{gather*}
Our weights are constructed based on detection functions:
\begin{gather*}
\omega_1 = \sigma_1,
\quad
\omega_2(x_1,x_2) = 1 + \cos\big((\pi/2)(x_1 + x_2)\big).
\end{gather*}
and $\omega_3$ is an estimated scale function as in Example 1.
We expect that the choice $\omega_1=\sigma_1$ will provide the largest
power for detecting $\sigma_1$. We also consider the choice $\omega_2$
to illustrate the test performance when we choose some reasonable
non-constant detection function. The detection function $\omega_3$ is
based on a locally cubic estimator for the scale function.

We conducted simulations consisting of 1000 runs. The results 
for $\omega_1$, $\omega_2$ and $\omega_3$ are given in Tables \ref{table2w1}, 
\ref{table2w2} and \ref{table2w3}, respectively. 
The figures corresponding to the test level $\alpha=5\%$ ($\sigma_0 \equiv 1$) 
and the fixed detection functions $\omega_1$ and $\omega_2$ show the tests
using the asymptotic quantile are conservative, which mirrors the
results from Example 1. At sample size $300$, the rejection rates for
the tests using $\omega_1$ and $\omega_2$ are both near $2\%$. The
test using $\omega_3$, an estimated scale function, generally produces
higher rejection rates (but still near the nominal level). At
the sample size $300$ the rejection rate for the test using $\omega_3$
is about $7\%$.

When we consider the remaining figures corresponding to the powers of
each test, we find considerable differences. The test using $\omega_1
= \sigma_1$ (Table \ref{table2w1}) provides, as expected, the best
results when $\sigma_1$ is in fact the underlying scale function. 
The corresponding figures for the test that uses the estimated scale function 
$\omega_3$ (Table \ref{table2w3}) are similar. The results in
Table \ref{table2w2} indicate that the test using $\omega_2 =
1 + \cos((\pi/2)(x_1 + x_2))$ is less effective for detecting
$\sigma_1$, but still quite good for the larger sample sizes 200 and 300.
Comparing the three tests in Tables \ref{table2w1}-\ref{table2w3}
for detecting $\sigma_2$, we see that only the test with the estimated scale 
function $\omega_3$ appears to be powerful.

Only at very large sample sizes can we expect that all three testing
procedures will provide similar results. In conclusion, we find the
test using an arbitrary non-constant weight function is useful but
will normally be outperformed by a test using estimated weights.

\begin{remark}[curse of dimensionality] \label{remCurseDim}
The simulation results in Tables \ref{table2w1}-\ref{table2w3}
suggest that our proposed tests, which are based on local polynomial
smoothers, may not be very reliable when the dimension of the
covariate vector is large. In this case the smoother (as well as other
nonparametric function estimators) will be affected by the ``curse of
dimensionality'', which is implicated by the entropy results in
Section \ref{aux}. To meet the situation with many covariates in
practice, we recommend working with dimension-reducing
transformations: choose, for example, a transformation $\xi$ of the
covariate vector $X$ such that $V = \xi(X)$ is just one covariate (and
the function estimator is not affected by the dimensionality
problem). A popular example would be the single-index model, where
$\xi$ is a linear combination of the components of $X$. Working with
such transformations will not change the independence-dependence
structure between the detection function and the errors, which is key
for our procedure to work.
\end{remark}



\section{Technical details} \label{aux}
In this section we present the proof of Theorem \ref{thmTn}, 
the limiting distribution of $T_n$ under the null hypothesis,
and some auxiliary results. 
As explained in Section \ref{marmodel}, we do not have to prove 
Theorem \ref{marmodel} for the MAR model:
it suffices to consider the full model and the test statistic $T_n$. 
Our approach consists of two steps. Our first step will be to use 
Theorem 2.2.4 in Koul's 2002 book on weighted empirical processes
to obtain the limiting distribution of an asymptotically linear
statistic (a sum of i.i.d.\ random variables) that is related to
$T_n$. Then we review some results from M\"uller, Schick and
Wefelmeyer (2009), who propose local polynomial smoothers to estimate
a regression function of many covariates. Using these results, we will
show that the statistic $T_n$ and the asymptotically linear statistic
are indistinguishable for large samples, i.e.\ they have the same 
limiting distribution.

The asymptotically linear statistic, which is an empirical process
related to $T_n$, is defined similarly to $T_n$ as
\begin{equation} \label{aslinstat}
\sup_{t \in \R} \Big| n^{-1/2}\sj W_j
 \Big\{ \1[\ve_j \leq t] - F(t) \Big\}
 \Big|
=
\sup_{t \in \R} \Big| n^{-1/2}\sj W_j
 \Big\{ \1[\sigma_0\err_j \leq t] - F(t) \Big\}
 \Big|,
\end{equation}
where $\ve_j = \sigma_0 \err_j$ is the unobserved ``model error'' from 
the null hypothesis and $W_1,\ldots,W_n$ are the 
standardised weights given in \Ref{Wj}. We
will now demonstrate that (under $H_0$) the requirements for Koul's theorem 
are satisfied. The asymptotic statement is given afterwards 
in Corollary \nolinebreak \ref{coraslinstat}.

Theorem 2.2.4 of Koul (2002) states that 
\begin{equation*}
\zeta_n(t) =  n^{-1/2}\sj D_j \Big\{
 \1\big[C_j \leq t\big] - K(t) \Big\}
 \xrightarrow[]{D} \xi \Big\{ B_0 \circ K(t) \Big\},
 \qquad t \in \R,\text{ as } n \to \infty,
\end{equation*}
where $B_0$ is the standard Brownian bridge in the Skorohod space
$D[0,1]$, independent of a random variable $\xi$. The roles of his
random variable $C_j$ and the square integrable random variable $D_j$,
which are assumed to be independent, are now played by $\sigma_0 \err_j$
and $W_j$, $j=1,\ldots,n$. The distribution function $K$ corresponds
to our error distribution function $F$ and is assumed to have a
uniformly continuous Lebesgue density. The random variable $\xi$ from
above comes from Koul's requirement that
\begin{equation*}
\Big|\frac{1}{n}\sj D_j^2 \Big|^{1/2} = \xi + \op
 \qquad \text{for some positive r.v. $\xi$}.
\end{equation*}
Here we work with $W_j$, in place of $D_j$, with $E(W_j^2) =
1$. Therefore, by the law of large numbers, 
$n^{-1}\sj W_j^2 = 1 + \op$ and, using the continuous mapping 
theorem, $|n^{-1}\sj W_j^2|^{1/2} = 1 + \op$, 
i.e.\ $\xi \equiv 1$. Hence we have
\begin{equation*}
n^{-1/2}\sj W_j \Big\{
 \1[\sigma_0\err_j \leq t] - F(t) \Big\}
 \xrightarrow[]{D} B_0 \circ F(t),
 \qquad t \in \R, \text{ as } n \to \infty.
\end{equation*}
Taking the supremum with respect to $t \in \R$, the right-hand side
becomes $\sup_{t \in \R} |B_0 \circ F(t)| = \sup_{t \in [0,1]}
|B_0(t)|$, which specifies the asymptotic distribution of the
asymptotically linear statistic \Ref{aslinstat}. Note that Koul's
theorem also provides the limiting distribution for a shifted version
$\hat \zeta_n$ of $\zeta_n$ that involves random variables $Z_1,
\ldots, Z_n$. Since we only need the simpler result for $\zeta_n$, we
do not need to verify the more complicated assumptions regarding the
$Z_j$'s. This shows the conditions of Theorem 2.2.4 in Koul (2002) are
indeed satisfied. We will formulate this result as a corollary. Since
we only require the weights to be square-integrable functions of $X_j$
with $E(W_j^2) = 1$, we will not require the explicit form \Ref{Wj}.

\begin{corollary} \label{coraslinstat}
Consider the homoskedastic nonparametric regression model 
$Y = r(X) + \sigma_0\err$. 
Assume the distribution function $F$ of the errors has 
a uniformly continuous Lebesgue density $f$ that is positive 
almost everywhere. Further, let $W_j$ be a square integrable 
function of $X_j$ satisfying $E(W_j^2) = 1$, $j=1,\ldots,n$. 
Then
\begin{gather} \label{coraslinstatresult}
\sup_{t \in \R} \Big| n^{-1/2}\sj W_j \Big\{
 \1[\sigma_0\err_j \leq t] - F(t) \Big\}
 \Big| \xrightarrow[]{D} \sup_{t \in [0,1]} |B_0(t)|,
 \qquad\text{as } n \to \infty,
\end{gather}
where $B_0$ denotes the standard Brownian bridge.
\end{corollary}

For our second step, we will show that $T_n$ and the asymptotically
linear statistic \Ref{aslinstat} are asymptotically equivalent. To
begin we rewrite $T_n$, using the identity (under $H_0$) $\hve = Y -
\rhat(X) = \sigma_0\err - \rhat(X) + r(X)$, as
\begin{equation*}
\sup_{t \in \R}\Big| n^{-1/2}\sj \hW_j\1\big[\hve_j \leq t\big]\Big|
 = \sup_{t \in \R}\Big| n^{-1/2}\sj 
 \hW_j\1\big[\sigma_0\err_j \leq t + \rhat(X_j) - r(X_j)\big] \Big|.
\end{equation*}
We will first consider the shift in the indicator function from 
$t$ to $t + \rhat - r$, which comes in because $T_n$ involves
an estimator $\hat r$ of the regression function. 

Consider now the H\"older space $H(d,\gamma)$ from Section
\ref{fullmodel}, i.e.\ the space of functions that have partial
derivatives of order $d$ that are H\"older with exponent $\gamma \in
(0,1]$. For these functions we define the norm
\begin{equation*}
\|h\|_{d,\gamma} = \max_{i \in I(d)} \sup_{x \in [0,1]^m}
 \big|D^ih(x)\big| + \max_{i \in I(d)}
 \sup_{x,\,y \in [0,1]^m,~ x \neq y}
 \frac{|D^ih(y) - D^ih(x)|}{\|x - y\|^\gamma},
\end{equation*}
where $\|v\|$ is the Euclidean norm of a real-valued vector $v$ and
\begin{equation*}
D^ih(x) = \frac{\partial^{i_1 + \cdots + i_m}}{
 \partial x_1^{i_1} \cdots \partial x_m^{i_m}} h(x),
 \qquad x = (x_1, \ldots, x_m) \in [0,1]^m.
\end{equation*}
Write $H_1(d,\gamma)$ for the unit ball of $H(d,\gamma)$ using this
norm.
These function spaces are particularly useful for studying local
polynomial smoothers $\rhat$ as defined in Section
\ref{fullmodel}. M\"uller et al.\ 
(2009) make use of
these spaces to derive many useful facts concerning regression
function estimation using local polynomials. We will use some of their
results to prove Theorem \ref{thmTn}; see Lemma \ref{lemrhat} below.
%
%
\begin{lemma}[\sc Lemma 1 of M\"uller, Schick and Wefelmeyer, 2009] 
\label{lemrhat}
Let the local polynomial smoother $\rhat$, the regression function
$r$, the covariate distribution $G$ and the error distribution $F$
satisfy the assumptions of Theorem \ref{thmTn}. Then there is a random
function $\ahat$ such that, for some $\alpha > 0$, 
\begin{gather} 
\label{ahatinH1}
 P(\hat a \in H_1(m,\alpha)) \to 1, \\
\label{rhatrahatneg}
 \sup_{x \in [0,1]^m} \big|\rhat(x) - r(x) - \ahat(x)\big| = \opn.
\end{gather} 
\end{lemma}

We now use these results to show the difference between the
asymptotically linear statistic \Ref{aslinstat} and an empirical
process related to the shifted version of $T_n$ (called $R_1$ in Lemma
\ref{lemTnModulus} below) are asymptotically negligible. 
An unweighted version of that difference 
(with $W_j =1$) is considered in 
Theorem 2.2 of M\"uller, Schick and Wefelmeyer (2007). 
Since that statistic does not involve centred weights, 
the second part of $R_1$ (called $R_2$ in the lemma) is not asymptotically
negligible: it becomes a stochastic drift parameter that depends on the
error density $f$ ($f(t) \int \hat a \, dQ$ in that article) and is therefore 
{\em not} distribution free. This is in contrast to our case where we have 
mean zero weights, so $R_2$ does not affect the limiting distribution.
%
%
%
\begin{lemma} \label{lemTnModulus}
Let the null hypothesis hold. Suppose the assumptions of Theorem
\ref{thmTn} on $\hat r$, $r$, $G$ and $F$ are satisfied. Let $W_j$ be
a square integrable function of $X_j$ satisfying $E[W_j^2] < \infty$,
$j=1,\ldots,n$. Then $\sup_{t \in \R} |R_1| = \opn$, where
\begin{equation*}
R_1 = \avj W_j\Big\{
 \1\big[\sigma_0\err_j \leq t + \rhat(X_j) - r(X_j)\big]
 - \1[\sigma_0\err_j \leq t] - F\big(t + \rhat(X_j) - r(X_j)\big)
 + F(t)\Big\}.
\end{equation*}
If, additionally, $E[W_j] = 0$, $j=1,\ldots,n$, then $\sup_{t \in \R}
|R_2| = \opn$, where
\begin{equation*}
R_2 = \avj W_j\Big\{
 F\big(t + \rhat(X_j) - r(X_j)\big) - F(t)\Big\}.
\end{equation*}
\end{lemma}
%
%
\begin{proof}
Observe that the class of functions
\begin{equation*}
\Fclass = \Big\{(X,\sigma_0\err) \mapsto W\Big\{
 \1\big[\sigma_0\err \leq t + a(X)\big] - F\big(t + a(X)\big) \Big\}
 \,:\,t\in\R,~a\in H_1(m,\alpha)\Big\}
\end{equation*}
is $G \otimes F$-Donsker, which follows from the fact that $W_j$ is a
fixed, square-integrable function of $X_j$ and the class of indicator
functions in the definition of $\Fclass$ is $G \otimes F$-Donsker
from Theorem 2.1 of M\"uller et al.\ (2007).
It then follows from Corollary 2.3.12 of van der Vaart and Wellner
(1996) that empirical processes ranging over the Donsker class
$\Fclass$  are asymptotically equicontinuous, i.e.\ we have, for any
$\varphi > 0$,
\begin{equation} \label{equicon}
\lim_{\kappa \downarrow 0}\limsup_{n \to \infty} 
P\Big(\sup_{\{\fmem_1,\,\fmem_2 \in \Fclass\,:\,\Var(\fmem_1 - \fmem_2) < \kappa\}}
 n^{-1/2}\Big|\sj \Big\{
 \fmem_1(X_j,\sigma_0\err_j) - \fmem_2(X_j,\sigma_0\err_j) \Big\}
 \Big| > \varphi\Big) = 0.
\end{equation}
We are interested in the case that involves the approximation $\hat a$ 
of $\hat r - r$ in place of $a$ (see Lemma \ref{lemrhat}). Then the 
corresponding class of functions is, in general, no longer Donsker (and 
the equicontinuity property does not hold). However, we can assume that 
$\hat a$ is in $H_1(m, \alpha)$, which holds on an event that has probability 
tending to one. This together with a negligibility condition on the
variance guarantees that the extended class of processes involving $\hat a$ 
is also equicontinuous. 

The term $R_1$ from the first assertion can be written as the sum of
\begin{equation} \label{modahatzero}
\avj W_j\Big\{
 \1\big[\sigma_0\err_j \leq t + \ahat(X_j)\big]
 - \1[\sigma_0\err_j \leq t] 
 - F\big(t + \ahat(X_j)\big) + F(t)\Big\}
\end{equation}
and
\begin{align} \label{modrhatrahat}
&\avj W_j\Big\{
 \1\big[\sigma_0\err_j \leq t + \rhat(X_j) - r(X_j)\big]
 - F\big(t + \rhat(X_j) - r(X_j)\big)\Big\} \\ \nonumber
&- \avj W_j\Big\{
 \1\big[\sigma_0\err_j \leq t + \ahat(X_j)\big]
 - F\big(t + \ahat(X_j)\big) \Big\}.
\end{align}
The first assertion, $\|R_1\|_{\infty} = \opn$, will follow
if we show this separately for the two terms in \eqref{modahatzero} 
and \eqref{modrhatrahat}. Consider \eqref{modahatzero} first. We
fix the function $\ahat$ by conditioning on the observed data 
$\Dset = \{(X_1,Y_1),\ldots,(X_n,Y_n)\}$.
The variance of a function from the extension of $\Fclass$ that involves
$\ahat$ instead of $a$ is
\begin{align*}
&\Var\Big[W\Big\{\1\big[\sigma_0\err \leq t + \ahat(X)\big]
 - \1[\sigma_0\err \leq t] - F\big(t + \ahat(X)\big)
 + F(t)\Big\}\,\Big|\,\Dset\Big] \\
&= E\Big[W^2\Big\{
 F\big(\max\{t,\,t+\ahat(X)\}\big)
 - F\big(\min\{t,\,t+\ahat(X)\}\big) \\
&\phantom{= E\Big[W^2\Big\{}
 - \Big\{F\big(\max\{t,\,t+\ahat(X)\}\big)
 - F\big(\min\{t,\,t+\ahat(X)\}\big)\Big\}^2
 \Big\}\,\Big|\,\Dset\Big].
\end{align*}
If this variance is $\op$, then the extended class of processes is equicontinuous,
and the term in \eqref{modahatzero} has the desired order $\opn$, uniformly in 
$t \in \R$.
That the variance condition holds true is easy to see: the last term
is bounded by $\|f\|_{\infty}E\big[W^2\big]\|\ahat\|_{\infty} = \op$ with
$\|\ahat\|_{\infty} = \op$ (see page 961 of the proof of Lemma 1 in
M\"uller et al., 2009). 

Turning our attention now to the second term \eqref{modrhatrahat}, we
have that $\rhat - r = (\rhat - r - \ahat) + \ahat$, and, by property
\Ref{rhatrahatneg} of Lemma \ref{lemrhat}, $A_n = \|\rhat - r -
\ahat\|_{\infty} = \opn$.  
Write $W_j^- = W_j\1[W_j < 0]$ and $W_j^+ = W_j\1[W_j \geq 0]$
for the negative and the positive part of $W_j$, i.e.\
$W_j = W_j^- + W_j^+$, $j=1,\ldots,n$. This yields
the following bounds for the weighted indicator functions:
\begin{equation*}
W_j^-\1\big[\sigma_0\err_j \leq t + \rhat(X_j) - r(X_j)\big]
 \leq W_j^-\1\big[\sigma_0\err_j \leq t - A_n + \ahat(X_j)\big],
\end{equation*}
\begin{equation*}
W_j^-\1\big[\sigma_0\err_j \leq t + \ahat(X_j)\big]
 \geq W_j^-\1\big[\sigma_0\err_j \leq t + A_n + \ahat(X_j)\big],
\end{equation*}
\begin{equation*}
W_j^+\1\big[\sigma_0\err_j \leq t + \rhat(X_j) - r(X_j)\big]
 \leq W_j^+\1\big[\sigma_0\err_j \leq t + A_n + \ahat(X_j)\big].
\end{equation*}
and
\begin{equation*}
W_j^+\1\big[\sigma_0\err_j \leq t + \ahat(X_j)\big]
 \geq W_j^+\1\big[\sigma_0\err_j \leq t - A_n + \ahat(X_j)\big].
\end{equation*}
Straightforward calculations show that \eqref{modrhatrahat} is bounded by
\begin{align*}
&\avj \Big\{ W_j^+ - W_j^- \Big\} \Big\{
 \1\big[\sigma_0\err_j \leq t + A_n + \ahat(X_j)\big]
 - F\big(t + A_n + \ahat(X_j)\big) \Big\} \\
&\quad - \avj \Big\{ W_j^+ - W_j^- \Big\} \Big\{
 \1\big[\sigma_0\err_j \leq t - A_n + \ahat(X_j)\big]
 - F\big(t - A_n + \ahat(X_j)\big) \Big\} \\
&\quad + \avj \Big\{ W_j^+ - W_j^- \Big\} \Big\{
 F\big(t + A_n + \ahat(X_j)\big)
 - F\big(t - A_n + \ahat(X_j)\big) \Big\} \\
&\quad - \avj W_j \Big\{
 F\big(t + \rhat(X_j) - r(X_j)\big)
 - F\big(t + \ahat(X_j)\big) \Big\} \\
&= \avj \big|W_j\big|\Big\{
 \1\big[\sigma_0\err_j \leq t + A_n + \ahat(X_j)\big]
 - F\big(t + A_n + \ahat(X_j)\big)\Big\} \\
&\quad - \avj \big|W_j\big|\Big\{
 \1\big[\sigma_0\err_j \leq t - A_n + \ahat(X_j)\big]
 - F\big(t - A_n + \ahat(X_j)\big)\Big\} \\
&\quad + \avj \big|W_j\big|\Big\{
 F\big(t + A_n + \ahat(X_j)\big)
 - F\big(t - A_n + \ahat(X_j)\big)\Big\} \\
&\quad - \avj W_j\Big\{
 F\big(t + \rhat(X_j) - r(X_j)\big)
 - F\big(t + \ahat(X_j)\big)\Big\}.
\end{align*}
Hence \eqref{modrhatrahat} is $\opn$ uniformly in $t \in
\R$ holds if we show
\begin{align} \label{modAn}
\sup_{t \in \R} \Big|
 &\avj \big|W_j\big|\Big\{
 \1\big[\sigma_0\err_j \leq t + A_n + \ahat(X_j)\big]
 - F\big(t + A_n + \ahat(X_j)\big)\Big\} \\ \nonumber
&- \avj \big|W_j\big|\Big\{
 \1\big[\sigma_0\err_j \leq t - A_n + \ahat(X_j)\big]
 - F\big(t - A_n + \ahat(X_j)\big)\Big\}
 \Big| = \opn,
\end{align}
\begin{equation} \label{modAnRem}
\sup_{t \in \R} \avj \big|W_j\big|\Big\{
 F\big(t + A_n + \ahat(X_j)\big)
 - F\big(t - A_n + \ahat(X_j)\big)\Big\}
 = \opn
\end{equation}
and
\begin{equation} \label{modFrhatrahatRem}
\sup_{t \in \R} \Big| \avj W_j\Big\{
 F\big(t + \rhat(X_j) - r(X_j)\big)
 - F\big(t + \ahat(X_j)\big)\Big\} \Big|
 = \opn.
\end{equation}

Beginning with \Ref{modAn}, since the random variables
$|W_1|,\ldots,|W_n|$ are square integrable, the class of functions
\begin{equation*}
\Fclass^+ = \Big\{
 (X,\sigma_0\err) \mapsto |W| \Big\{
 \1\big[\sigma_0\err \leq t + a(X)\big]
 - F\big(t + a(X)\big) \Big\}
 \,:\,t \in \R,~a \in H_1(m,\alpha) \Big\}
\end{equation*}
is also $G \otimes F$-Donsker. Therefore the asymptotic
equicontinuity property holds for empirical processes ranging over
$\Fclass^+$, i.e.\ \Ref{equicon} holds with $\Fclass^+$ in place of
$\Fclass$. However, rather than investigating the situation where
$\ahat$ is limiting toward zero, as we did above, we will consider two
sequences of real numbers $\{s_n\}_{n = 1}^{\infty}$ and  $\{t_n\}_{n
  = 1}^{\infty}$ satisfying $|t_n - s_n| = o(1)$, which corresponds to
the case of random sequences $t \pm A_n$ conditional on the data
$\Dset$. Analogously to the calculations following \Ref{equicon}, we
now prove the variance condition 
for the function $(X,\sigma_0\err) \mapsto |W|\{\1[\sigma_0\err \leq
t_n + a(X)] - \1[\sigma_0\err \leq s_n + a(X)] - F(t_n + a(X)) +
F(s_n + a(X))\}$. The variance is
\begin{align*}
&\Var\Big[|W|\Big\{
 \1\big[\sigma_0\err \leq t_n + a(X)\big]
 - \1\big[\sigma_0\err \leq s_n + a(X)\big]
 - F\big(t_n + a(X)\big)
 - F\big(s_n + a(X)\big) \Big\}\Big] \\
&= E\Big[W^2\Big\{
 F\big(\max\{t_n + a(X),\,s_n + a(X)\}\big)
 - F\big(\min\{t_n + a(X),\,s_n + a(X)\}\big) \\
&\phantom{= E\Big[W^2\Big\{}
 - \Big\{F\big(\max\{t_n + a(X),\,s_n + a(X)\}\big)
 - F\big(\min\{t_n + a(X),\,s_n + a(X)\}\big)\Big\}^2
 \Big\}\Big].
\end{align*}
and bounded by $\|f\|_{\infty}E[W^2]|t_n - s_n| = o(1)$. Hence we have
equicontinuity, and therefore, for any $a \in H_1(m,\alpha)$
and sequences of real numbers $\{s_n\}_{n = 1}^{\infty}$ and
$\{t_n\}_{n = 1}^{\infty}$ satisfying $|t_n - s_n| = o(1)$,
\begin{align*}
&\sup_{t \in \R} \Big|
 \avj \big|W_j\big|\Big\{
 \1\big[\sigma_0\err_j \leq t_n + a(X_j)\big]
 - F\big(t_n + a(X_j)\big) \Big\} \\
&\phantom{\sup_{t \in \R} \Big|}
 - \avj \big|W_j\big|\Big\{
 \1\big[\sigma_0\err_j \leq s_n + a(X_j)\big]
 - F\big(s_n + a(X_j)\big) \Big\} \Big| = \opn.
\end{align*}
Equation \Ref{modAn} follows analogously,
with $t_n$ replaced by
$t+A_n$, $s_n$ by $t-A_n$ and $a$ by $\hat a$.
%
%

Now consider \Ref{modAnRem} and \eqref{modFrhatrahatRem}.
Since $E|W| \leq E^{1/2}[W^2] < \infty$, $n^{-1}\sj|W_j|$ is consistent
for $E|W|$. The left-hand side of \Ref{modAnRem}
is bounded by $2\|f\|_{\infty}A_nn^{-1}\sj|W_j|$ and
\eqref{modFrhatrahatRem} is bounded by
$\|f\|_{\infty}A_nn^{-1}\sj|W_j|$. Since $A_n = \opn$, these bounds 
are also $\opn$, i.e.\ \Ref{modAnRem} and
\eqref{modFrhatrahatRem} hold. This implies that the term in
\eqref{modrhatrahat} has order $\opn$, uniformly in $t \in \R$, which completes 
the proof of $\|R_1\|_{\infty} = \opn$.

We will now prove the second assertion that $\|R_2\|_{\infty} =
  \opn$. The proof is simpler than the proof of the first assertion
since we now require that the random variables $W_1,\ldots,W_n$ have mean 
zero, which allows us to use the central limit theorem.
Write $R_2$ as
\begin{align*}
R_2 &= \avj W_j\Big\{ F\big(t + \ahat(X_j)\big) - F(t)
 - E\Big[F\big(t + \ahat(X)\big) - F(t)
 \,\Big|\,\Dset\Big] \Big\} \\
&\quad + E\Big[F\big(t + \ahat(X)\big) - F(t)
   \,\Big|\,\Dset\Big]\Big(\avj W_j\Big)
 \\ &\quad
 + \avj W_j\Big\{
 F\big(t + \rhat(X_j) - r(X_j)\big) - F\big(t + \ahat(X_j)\big)
 \Big\}.
\end{align*}
Then $\|R_2\|_{\infty}$ is bounded by three terms:
\begin{equation} \label{modFahat}
\sup_{t \in \R} \Big| \avj W_j \Big\{
 F\big(t + \ahat(X_j)\big) - F(t)
 - E\Big[F\big(t + \ahat(X)\big) - F(t)\,\Big|\,\Dset\Big]
 \Big\} \Big|,
\end{equation}
\begin{equation} \label{modFahatRem}
\sup_{t \in \R} \Big| E\Big[
 F\big(t + \ahat(X)\big) - F(t)\,\Big|\,\Dset\Big] \Big|
 \Big| \avj W_j \Big|,
\end{equation}
and the third term is the left-hand side of \Ref{modFrhatrahatRem},
which we have already shown is $\opn$. From the arguments above, it
follows for the class of functions
\begin{equation*}
\Fclass_2 = \Big\{ X \mapsto W \Big\{
 F\big(t + a(X)\big) - E\big[F\big(t + a(X)\big)\big]
 \Big\}\,:\,t \in \R,~a \in H_1(m,\alpha) \Big\}
\end{equation*}
to be $G$-Donsker. Therefore, empirical processes ranging over
$\Fclass_2$ are asymptotically equicontinuous as in \Ref{equicon}, but
now without $\sigma_0 e$ and with $\Fclass_2$ in place of $\Fclass$.

As before, we can assume that $\ahat$ belongs to
$H_1(m,\alpha)$. We will now show the variance condition is satisfied
for the function $X \mapsto W\{F(t + \ahat(X)) - F(t) - E[F(t +
\ahat(X)) - F(t)\,|\,\Dset]\}$. This variance is equal to
\begin{align*}
&E\Big[W^2\Big\{F\big(t + \ahat(X)\big) - F(t)\Big\}^2
 \,\Big|\,\Dset\Big]
 + E\big[W^2\big]E^2\Big[
 F\big(t + \ahat(X)\big) - F(t)\,\Big|\,\Dset\Big] \\
&- 2E\Big[W^2 \Big\{ F\big(t + \ahat(X)\big) - F(t)
 \Big\}\,\Big|\,\Dset\Big]
 E\Big[F\big(t + \ahat(X)\big) - F(t)\,\Big|\,\Dset\Big],
\end{align*}
and is bounded by $2\|f\|_{\infty}^2
  E[W^2]\|\ahat\|_{\infty}^2$. Since 
$\|\ahat\|_{\infty} = \op$, the bound above is $\op$ and the
variance is asymptotically negligible. Hence we have
equicontinuity and \Ref{modFahat} is $\opn$, as desired.

Finally we can bound \Ref{modFahatRem} by
$\|f\|_{\infty}\|\ahat\|_{\infty}|n^{-1}\sj W_j|$.
The central limit theorem combined with $E[W_j] = 0$, $j=1,\ldots,n$,
gives $|n^{-1}\sj W_j| = \Opn$. Since $\|\ahat\|_{\infty} = \op$,
both the bound above and \Ref{modFahatRem} are of the order
$\opn$. This completes the proof of the second assertion that
$\|R_2\|_{\infty} = \opn$.
\end{proof}

Using the results of Lemma \ref{lemTnModulus}, we will now show that
the test statistic $T_n$ and the asymptotically linear statistic above
are asymptotically equivalent. This will imply the limiting
distribution of $T_n$ is the same as that of the asymptotically linear
statistic \Ref{aslinstat}, which we have already investigated; see Corollary
\ref{coraslinstat}.

\begin{proof}[Proof of Theorem \ref{thmTn}]
Consider the asymptotically linear statistic from \Ref{aslinstat},
\begin{equation*}
n^{-1/2}\sj W_j \Big\{ \1[\sigma_0\err_j \leq t] - F(t) \Big\},
\end{equation*}
with $W_j$ given in \Ref{Wj}. It follows, by the arguments preceding
Corollary \ref{coraslinstat}, for this statistic to have the limiting 
distribution $B_0 \circ F(t)$, where $B_0$ is the Brownian bridge.
We will now show that
\begin{equation} \label{diff}
\sup_{t \in \R} \Big| \avj \hW_j\1\big[\hve_j \leq t\big]
 - \avj W_j\Big\{\1[\sigma_0\err_j \leq t] - F(t)\Big\} \Big|
 = \opn.
\end{equation}
Combining the above, the desired statement of Theorem \ref{thmTn}
concerning the limiting distribution of the test statistic $T_n$
follows, i.e.\
\begin{equation*}
T_n = \sup_{t \in \R} \Big| n^{-1/2}\sj
 \hW_j\1\big[\hve_j \leq t\big] \Big| 
 \xrightarrow[]{D} \sup_{ t \in [0,1]} |B_0(t)|.
\end{equation*}
It follows from $\sj \hW_j = 0$ that we can decompose the difference 
in \Ref{diff} into the following sum of five remainder terms: 
$R_1 + R_3 + R_4 - R_5 - R_6$, where $R_1$ and $R_2$  
(which is part of $R_3$) are the remainder terms of Lemma 
\ref{lemTnModulus}, and where the other terms are defined as follows,
\begin{align*}
R_3 &=
\hat V R_2, \qquad
\hat V = \Big( \frac{\Var[\omega(X_1)]}{
 \avj \{\omega(X_j) - \avk \omega(X_k)\}^2} \Big)^{1/2},
\\
 R_4 &= (\hat V-1)
 \Big(\avj W_j \Big\{
 \1\big[\sigma_0\err \leq t + \rhat(X_j) - r(X_j)\big]
 - F\big(t + \rhat(X_j) - r(X_j)\big) \Big\}\Big),
\\
R_5 &= \hat V
\Big(\avj W_j\Big)
\Big( \avj \Big\{
 \1\big[\sigma_0\err \leq t + \rhat(X_j) - r(X_j)\big]
 - F\big(t + \rhat(X_j) - r(X_j)\big) \Big\} \Big),
\\
 R_6 &= \hat V
\Big(\avj W_j\Big)
\Big(
 \avj \Big\{ F\big(t + \rhat(X_j) - r(X_j)\big) - F(t) \Big\}
 \Big).
\end{align*}
Showing $\sup_{t \in \R} |R_{i}| = \opn$, $i=1,\ldots,6$,
will conclude the proof. The statement for $i = 1$ holds true by
the first part of Lemma \ref{lemTnModulus} and the statement for $i =
2$ holds true by the second part of the same lemma. Note that the
assumptions of both statements of Lemma \ref{lemTnModulus} are
satisfied for our choice of weights $W_1,\ldots,W_n$. The statement
for $i = 3$ follows from the second statement of the same lemma
regarding $R_2$ and from the fact that the first quantity of $R_3$,
$\hat V$, is a consistent estimator of one.

To show $\sup_{t \in \R}|R_4| = \opn$, we only need to demonstrate that 
\begin{equation} \label{R2}
\sup_{t \in \R}\Big| \avj W_j \Big\{
 \1\big[\sigma_0\err_j \leq t + \rhat(X_j) - r(X_j)\big]
 - F\big(t + \rhat(X_j) - r(X_j)\big) \Big\} \Big|
 = \Opn,
\end{equation}
because the first term of $R_4$ both does not depend on $t$ and is
asymptotically negligible. 
To verify \Ref{R2}, combine the statement
for $R_1$ with the limiting result \Ref{coraslinstatresult} from
Corollary \ref{coraslinstat} for the asymptotically linear statistic,
which shows $n^{-1} \sj W_j \{ \1[\sigma_0\err_j \leq t] - F(t) \} =
O_p(n^{-1/2})$, uniformly in $t \in \R$.

Now consider $R_5$ and remember that both Corollary \ref{coraslinstat}
and the first statement of Lemma \ref{lemTnModulus} cover the special
case where all of the weights are equal to one, i.e.\ \Ref{R2} holds
with $W_j=1$, $j=1,\ldots, n$. Therefore, the third term of $R_5$
is $\Opn$, uniformly in $t \in \R$. It is clear for the product of the
first and second terms of $R_5$ to be $\op$. It then follows that
$\sup_{t \in \R}|R_5| = \opn$.

We find that $\sup_{t \in \R} |R_6|$ is bounded by
$$
  \sup_{t \in \R} |f(t)| \, \hat V \,
  \Big( \sup_{x \in [0,\,1]^m} \big|\ahat(x)\big|
 + \sup_{x \in [0,\,1]^m} \big|\rhat(x) - r(x) - \ahat(x)\big| \Big)
  \Big|\avj W_j\Big|,
$$
with $\hat V$ a consistent estimator of one. As
in the proof of Lemma \ref{lemTnModulus}, we use $\sup_{x \in
  [0,\,1]^m} |\ahat(x)| = \op$ and $\sup_{x \in [0,\,1]^m} |\rhat(x) -
r(x) - \ahat(x)| = \op$, e.g.\ see property \Ref{rhatrahatneg} of
Lemma \ref{lemrhat}. Hence, the third term in the bound above is
$\op$. We can apply the central limit theorem to treat the fourth
quantity and find it is $\Opn$. Combining these findings yields the
bound above is $\opn$. This implies $\sup_{t \in \R} |R_6| = \opn$.
\end{proof}


\section{Concluding remarks} \label{conrem}

We have introduced a completely nonparametric test to detect 
heteroskedasticity in a regression model with multivariate covariates
that not only converges at the parametric root-$n$ rate, but is also
strikingly simple. 
The test has the advantage that it is asymptotically distribution free, 
i.e.\ quantiles are readily available. 
The same test can also be applied if responses are 
missing at random by simply omitting the cases that are not complete
and using the same quantiles.
Crucial for the performance of the test is the choice of weights: we have 
seen that the detection function $\omega$ should be highly correlated
with the scale function $\sigma$ to maximise the power of the test.
If $\omega$ and $\sigma$ are not or only vaguely correlated, 
then the test has no or almost no power. This suggests that it is best
to work with a (possibly estimated) detection function $\omega$ that 
has the same shape as $\sigma$.

The methodology developed in this article can be easily extended to form
related tests for other model conditions and/or for other regression models.
We will illustrate this below by means of two examples. In Example 
\ref{exampleParamScale} we show how we can modify our statistic \Ref{Tn} to 
obtain a test for the parametric form of the scale function. 
In Example \ref{exampleAdditive} 
we discuss a possible test for additivity of the regression
function. This example is representative for tests that are
based on detecting differences in the regression function under the null 
and under the alternative hypothesis.

\begin{example} \label{exampleParamScale}
Our method can be modified to obtain tests for the parametric form 
of the scale function, i.e.\ with null hypothesis
$\sigma(\cdot) = \sigma_\theta(\cdot)$ for some $\theta \in \R^p$.
Under the null hypothesis we have $\ve = \sigma_\theta(X) \err$,
with $\err$ scaled and centred as before,
so that the standardised residuals $Z = \ve/\sigma_\theta(X) = \err$ 
and $X$ are independent, which is the key prerequisite for our test
for heteroskedasticity. Hence we can simply use our test based on 
the statistic
$$
\sup_{t \in \R} \Big| n^{-1/2} \sj
\hW_j\1\big[\hat Z_j \leq t\big] \Big|, \quad
\hat Z_j = \frac{\hve_j}{\sigma_{\hat \theta}(X_j)},
$$
which is the statistic $T_n$ from \Ref{Tn}, now with estimated 
standardised residuals $\hat Z$ in place of $\hve$, 
where $\hat \theta$ is a consistent estimator of $\theta$.
We expect that the test will also be asymptotically distribution free: 
the standardisation will result in an asymptotically negligible drift which 
only affects the indicators ${\bf 1}[\hat Z_j \le t]$. This can be handled 
using similar arguments. 
\end{example}


\begin{example} \label{exampleAdditive}
Another important application are tests about the regression function.
One might, for example, want to check if certain components of the regression 
function are constant or irrelevant, or if the regression function has 
a specific structure. 
Suppose, for example, we assume an additive nonparametric model with 
two-dimensional covariate vector 
$X_j=(X_{1,j},X_{2,j})^\top$, i.e.\ the regression function is
$r(X_{1,j},X_{2,j}) = r_1(X_{1,j}) + r_2(X_{2,j})$ under the null hypothesis. 
The test statistic is
$$
\sup_{t \in \R} \Big| n^{-1/2} \sj
\hW_j\1\big[\hve_j \leq t\big] \Big|, \quad
\hve_j = Y_j - \hat r_1(X_{1,j}) - \hat r_2(X_{2,j}),
$$
where $\hat r_1(x_1)$ and $\hat r_2(x_2)$ estimate $\bar r_1(x_1)$ and 
$\bar r_2(x_2)$, with $\bar r_1 = r_1$ and $\bar r_2 = r_2$ under $H_0$.
For sufficiently large $n$ we have the approximation
$$
\1[\hve \le t] \approx \1[\ve \le t + s(X_1,X_2)], \quad
s(x_1,x_2) = \bar r_1(x_1) + \bar r_2(x_2)  - r(x_1,x_2),
$$
where the shift $s$ is zero if the null hypothesis holds true.
To understand the construction and the power of the test, consider
again \Ref{diffomega} from the introduction (cf.\ Remark 
\ref{remTnpow1} on ``power under fixed alternatives'').
For simplicity assume $\sigma(\cdot) = \sigma_0$, i.e.\ 
$\ve = \sigma_0 e$ and $X$ are independent. This time we have to take the
shift into account and consider the difference
$$
E\big[\big\{\omega(X) - E\big[\omega(X)\big]\big\} 
  \1[\ve \leq t + s(X_1,X_2)]\big],
$$
which is zero under $H_0$, due to the independence assumption and 
since $s \equiv 0$. Under the alternative hypothesis we have
$$
E\big[ \omega(X) \1[ \ve \leq t + s(X_1,X_2)] \big]
= 
E\Big[ \omega(X) F\Big(\frac{t + s(X_1,X_2)}{\sigma_0}\Big) \Big],
$$
which, in general, does not equal 
$E[\omega(X)] E\big(F[ \{t + s(X_1,X_2)\}/\sigma_0] \big)$,
i.e.\ the above difference is not zero. As already observed in
Remark \ref{remTnpow1}, we expect a good power if the detection
function is suitably chosen, here in such a way that $\omega(X)$ 
and the shift function $s(X_1, X_2)$ are highly correlated.
\end{example}

\section*{Acknowledgements}
Justin Chown acknowledges financial support from the contract 
`Projet d'Actions de Recherche Concert\'ees' (ARC) 11/16-039 of the 
`Communaut\'e fran\c{c}aise de Belgique', granted by the 
`Acad\'emie universitaire Louvain', the IAP research network P7/06 of
the Belgian Government (Belgian Science Policy) and the Collaborative
Research Center ``Statistical modeling of nonlinear dynamic
processes'' (SFB 823, Teilprojekt  C4) of the German Research
Foundation (DFG).


\end{document}